\documentclass[a4paper]{article}
\usepackage{color}
\usepackage{verbatim}
\usepackage{amsthm}
\usepackage{amsmath}
\usepackage{graphicx}
\usepackage{amssymb}
\usepackage{bbold}
\usepackage{esint}
\usepackage{subfigure}
\usepackage{color}
\usepackage[all]{xy}

\usepackage{anysize}
\marginsize{2cm}{2cm}{2cm}{2cm}

  \newtheorem{prop}{Proposition}

 \newtheorem{definition}{Definition}

\def\beq{\begin{equation}}
\def\eeq{\end{equation}}
\def\bea{\begin{eqnarray}}
\def\eea{\end{eqnarray}}

\makeatletter

\begin{document}

\title{2D Quantum Double Models From a 3D Perspective}
\author{Miguel Jorge Bernab\'{e} Ferreira$^{a}$\footnote{migueljb@if.usp.br}, ~Pramod Padmanabhan$^{a}$\footnote{pramod23phys@gmail.com},\\ Paulo Teotonio-Sobrinho$^{a}$\footnote{teotonio@if.usp.br}} 
\maketitle

\begin{center}
{Departmento de F\'{i}sica Matem\'{a}tica Universidade de S\~{a}o Paulo - USP} \\
{\small
CEP 05508-090 Cidade Universit\'aria, S\~{a}o Paulo - Brasil
}

\end{center}

\begin{abstract}
In this paper we look at 3D lattice models that are generalizations of the state sum model used to define the Kuperberg invariant of 3-manifolds. The partition function is a scalar constructed as a tensor network where the building blocks are tensors given by the
structure constants of an involutary Hopf algebra $\mathcal{A}$.
These models are very general and are hard to solve in its entire parameter space. One can obtain familiar models, such as ordinary gauge theories, by letting $\mathcal{A}$ be the group algebra $\mathbb{C}(G)$ of a discrete group $G$ and staying  on
a certain region  of the parameter space. We consider the transfer matrix of the
model and show that Quantum double Hamiltonians are derived from a particular choice of the parameters. Such a construction naturally leads to the star and plaquette operators of the quantum double Hamiltonians, of which the toric code is a special case when $\mathcal{A}=\mathbb{C}(\mathbb{Z}_2)$.
This formulation is convenient to study ground states of these generalized quantum double models where they can naturally be interpreted as tensor network states. For a surface $\Sigma$, the ground state degeneracy is determined by the Kuperberg 3-manifold invariant of $\Sigma\times S^1$. It is also possible to obtain extra models by simple enlarging the allowed parameter space  but keeping the solubility of the model. While some of these extra models have appeared before in the literature, our 3D perspective allows for an uniform description of them.

\end{abstract}

\section{Introduction}
~

In recent years numerous efforts have gone towards finding systems exhibiting topological order~\cite{WenBook}. Topological order is believed to classify new phases of matter which cannot be classified under Landau's symmetry breaking scheme of classifying the states of matter. Their existence came to be known after the discovery of the fractional quantum Hall effect (see reviews by~\cite{DSP, PG}) and high temperature superconductivity~ \cite{PWA}.
These have been spurred by the need to find viable realizations of protected qubits to be used in topological quantum computation~\cite{NR, MF, AK}. In such attempts to find systems with topological order, the exactly soluble 2D lattice models of Levin and Wen~\cite{LW} provide a general method for obtaining anyon models~\cite{MR, FG, FW} using modular tensor categories. They exhibit the most  general range of possible quasiparticle statistics of any known Hamiltonian lattice theory. Indeed, Hamiltonians of this kind are thought to exist for a large class of achiral anyon theories known as
quantum double models. This class encompasses many of the previously studied anyon lattice models such as the Toric code \cite{AK}, many of its generalizations \cite{CC, EFH, AFF} and doubled Chern-Simons
theories \cite{FNSWW}.

Chain-Mail link invariants~\cite{Robert} have been used to realize such models from the spacetime perspective~\cite{FS1, FS2}, realizing the partition function of the Levin-Wen model as a knot invariant of a complicated link in three dimensional space. This provides a more physical picture of the Levin-Wen lattice model in terms of Wilson loops living on the edges of the three dimensional spacetime lattice. As the Hamiltonian of this model is made of commuting projectors the spectrum of the model can be obtained easily. The quasi particles are given a interpretation in the spacetime picture by adding additional strings to the chain-mail link. The chain-mail link invariant is known to be equivalent to the Turaev-Viro state sum invariant of 3-manifolds. Alternatively these models can also be related to the Turaev-Viro invariant of a closed 3-manifold. The invariant in this case is built out of a tensor network with values in a spherical category for a triangulation of the 3-manifold~\cite{KuperKoni}.

While the above models construct the Levin-Wen models starting from the Turaev-Viro 3-manifold invariants, in this paper we do the analogue for the toric code and its generalizations called quantum double models. Our starting point, however, will be the Kuperberg invariant~\cite{kuperberg} of 3-manifolds. This invariant uses a involutory Hopf algebra and a discrete presentation of the manifold known as the Heegaard splitting. In order to make contact with discretizations actually used in physical models, we use Heegaard splittings that come from triangular or cubic lattices. Our presentation is therefore similar to the one used in~\cite{CFS} to construct the same invariant.

The data needed to define the Kuperberg invariant $\mathcal{F}(M)$ is a Heegaard decomposition of the 3-manifold $M$ and a involutory Hopf algebra $\mathcal{A}$. It can be casted as a state sum model
by using a Heegaard splitting coming from a triangulation of the manifold. From this point of view the Kuperberg invariant ${\cal F}$ is essentially a partition function constructed out of weights associated to links and faces of a triangulated 3-manifold. These weights can be decomposed into the structure constants of an involutory Hopf algebra $\mathcal{A}$, where the antipode map squares to the identity. It is useful to regard the multiplication ${m_{ab}}^c$, the co-multiplication ${\Delta_{ab}}^c$ and the antipode $S_a^b$ of the Hopf algebra as tensors in $\mathcal{A}$. The invariant  ${\cal F}$ is a scalar constructed as a tensor network built out of ${m_{ab}}^c, {\Delta_{ab}}^c$ and $S_a^b$.

Turaev-Viro invariant $Z(M)$ is also a partition function of a state sum model. The relation between the Turaev-Viro and Kuperberg state sum models can be understood as follows. In the particular case when $\mathcal{A}$ is the group algebra of a group $G$, the states to be summed over are group elements associated to the links of the triangulation and the corresponding lattice field theory takes the familiar form of a gauge theory. The Turaev-Viro state sum model on the other hand  starts from a fusion category $\mathcal{C}$. The degrees of freedom associated to links are elements of $\mathcal{C}$  and therefore they are different lattice models. However, if one takes $\mathcal{C}$ to be the category of representations of $\mathcal{A}$, it has been proved in \cite{barretwestbury} that  ${\cal F}(M)=Z(M)\textrm{dim}(\mathcal{A})$. There is a closer relation between the two state sum models for the particular case when $\mathcal{A}$ is a group algebra of a finite group $G$. For this case one can show that the Turaev-Viro state sum model is a dual description of the Kuperberg state sum model when written in terms of spin network states where links are labeled by irreducible representations. This can be derived by
using results discussed in \cite{prboliver} and \cite{kadar}.

Both Kuperberg and Turaev-Viro state sum models are 3D lattice topological QFTs. They describe the low energy limit of fully dynamical lattice models such as the quantum double models. One could ask what 3D lattice models reproduces (in the Hamiltonian formalism) not only the low energy states but the entire spectrum of these models. Such a 3D field theory can not be a lattice TQFT. We show that a 3D model that resembles a traditional lattice gauge theory is the answer to this question. Instead of constructing Hamiltonians with a given topological phase as it is done in the string net approach, the present work, however, goes in an opposite direction. We start from a generalization of familiar 3D lattice gauge theory with a certain parameter space and look for topological order. The 3D model defined on the paper is inspired by the state sum model used to define Kuperberg invariant but it is not topological. Our model and the Kuperberg state sum model share the same state space but it has a non-trivial  dynamics. It generalizes ordinary gauge theories in a way that a Hopf algebra plays an analogous role to the gauge group. If one chooses $\mathcal{A}$ to be the group algebra of a group $G$, ordinary Wilson lattice gauge theory corresponds to a curve in the parameter space. The quantum double model appears when the model is restricted to a 2D surface of the parameter space. The Kuperberg TQFT appears as a point in the parameter space. Going away from this 2D surface leads to deformations of the quantum double model. Some of these deformations are still exactly soluble lattice models and some are not.

A precise relation between the Turaev-Viro model, the quantum double model and the Levin-Wen models have also been studied from a mathematical point of view in \cite{arxiv1,arxiv2,arxiv3}. For the specific case of the quantum double models that concern us, they are able to set a correspondence that goes beyond the ground state by considering Turaev-Viro model on surfaces with boundaries. The particle excitations of the quantum double model correspond to punctures on the 2D surface. This correspondence can only be made at the level of Hilbert space of states since the Turaev-Viro model is topological whereas  the quantum double model has a nontrivial Hamiltonian. The 3D model considered in this paper, on the other hand, is not topological. In our case, the correspondence with quantum double models is also dynamic in the sense that the the quantum double Hamiltonian is the logarithm of transfer matrix of the model.

The quantum double model phase is not the only quantum phase present in the parameter space. We give an example of this fact in the simplest case of $\mathcal{A}=\mathbb{C}(\mathbb{Z}_2)$. We show that for a certain choice of parameters, the corresponding Hamiltonian is a modification of the usual toric code with the following features. The model is exactly soluble and have essentially the same particles as the toric code. However, some of the original dyons from the toric code are now bound states with zero energy. As a consequence, the Ground State Degeneracy (GSD) is not the same as the toric code. Parameters can be fixed such that GSD is increased by a factor proportional to the exponential of the area (number of plaquettes) of the surface. This is a simple example of a quantum phase that is described in the low energy limit not by a TQFT but by a quasi-TQFT. This example opens the question of classifying all quantum phases, not all of them topological, for the 3D  model for the case of a generic involutory Hopf algebra. Although we do not have a solution to this problem, we can nevertheless  expect some limitations in the list of topological phases of the model. That comes from the fact that more general algebras, such as non-involutory Hopf algebras and weak Hopf algebras, are excluded from the model in its present formulation. It means that not all topological phases classified by fusion category theory will be present. In particular, phases described by twisted quantum doubles as in \cite{arxiv4,yw}. However, the model not only embeds quantum double models, but gives rise to interesting deformations as one explores the parameter space.

Perturbations of the quantum double model have been considered before, for example in~\cite{Bombin, bur}. One of their motivations was to investigate phase transitions between different topological phases via a mechanism of charge condensations. The perturbations in section \ref{sec-beyond} are not meant to achieve such an effect but are intended to give examples of how to depart from the quantum double models while remaining in the parameter space of the 3D model. This parameter space, however, is large enough to accommodate other types of deviations from the quantum double model. A simple example is given in Section 7. As mentioned before, it leads to a quantum phase that is not exactly topological. A complete analysis of this kind of model, especially in the non-Abelian case, is beyond the scope of the present paper.

We now outline our formalism using Kuperberg invariants. The models are parametrized by an element $z$ in the center of $\mathcal{A}$ and an element $z^*$ in the center of the dual algebra $\mathcal{A}^*$. The partition function $Z(\mathcal{A},z,z^*)$ is proportional to the Kuperberg invariant ${\cal F}$ only in the limit $z\rightarrow \eta$, $z^*\rightarrow \epsilon$~\cite{n1, MJBV}. Here $\eta$ and $\epsilon$ denote the unit and co-unit of the Hopf algebra $\mathcal{A}$.
Such techniques are especially useful for quasi-topological field theories, which arise by relaxing some of the conditions required for the theory to be topological~\cite{PB, PY}. In addition, the observation that $Z(z,z^*)$ is a scalar constructed as a tensor network have been used too generalize the Kramers and Wannier dualities of  lattice models \cite{n2}. We have also observed that partition functions of classical statistical mechanical models like the 3D Ising model and the lattice gauge theories can be constructed in a similar fashion.

These results are for three dimensional lattice models. Here we extend this formalism by considering the three dimensional manifold $M$ as spacetime. In other words, $M$ is of the form $\Sigma \times I$ where $\Sigma$ is a 2D surface and $I=[0,\delta t]$ is a time interval. In particular, we consider a discretization of $M$ such that $I$ is one single lattice step. The tensor network analogue to $Z(\mathcal{A},z,z^*)$ is no longer a scalar since it will have one free leg for each link in the lattice $\Sigma\times\{0\}\cup \Sigma\times\{\delta t\}$. That is precisely the transfer matrix $U$ for the model. In order to relate $U$ with the quantum double models, however, we need to take into account the splitting of $M$ into space and time directions. The weights associated to both directions are not necessarily the same. Just as $Z(\mathcal{A},z,z^*)$ depends on two parameters, $U$ will be parametrized by a pair of elements $z_S,z_T$ in the center of $\mathcal{A}$ and another pair $z_S^*, z_T^*$ in the center of the dual $\mathcal{A}^*$. The labels $S$ and $T$ refer to space and time directions respectively.

The transfer matrix $U(\mathcal{A},z_T,z_S,z_T^*,z_S^*)$ is very general. It encompasses models such as ordinary lattice gauge theories with matter in the regular representation. It is necessary to look at particular subsets of the parameter space if one is interested in soluble models. In this paper we set $z_T=\eta$, $z^*_S=\epsilon$. The algebra $\mathcal{A}$ can be any involutory Hopf algebra, not necessarily the group algebra $\mathbb{C}(G)$. In this paper we derive some of the models that result from setting  $z_T\neq \eta$ or  $z^*_S\neq \epsilon$ without going, however, into much detail. These new models will be carefully analyzed in another paper \cite{p2}.

By conveniently splitting this transfer matrix at each link, we write down the transfer matrix $U$ as a product of operators acting on the vertices and operators acting on the spacelike plaquettes of $\Sigma$. In other words $$ U = \prod_v A_v(z_T^*)\prod_p B_p(z_S)$$ where $v$ and $p$ denotes a vertex and a plaquette respectively. By writing the transfer matrix $U$ as $U = e^{-H\delta t}$, we obtain the Hamiltonian $H$ by taking the logarithm of both sides, namely
$$ H= -~\gamma_p~\sum_p \left(\mathbf{1}-\frac{2}{n}~B_p^0 \right) -~\gamma_s~\sum_s\left(\mathbf{1}-\frac{2}{n}~A_s^0\right) -\gamma ~\mathbf{1}\;, $$
where $\mathbf{1}$ is the identity matrix. The vertex operator $A^0_v=A_v(\epsilon)$ and the plaquette operator $B^0_p=B_p(\eta)$ are projectors and are precisely the ones occurring in the Hamiltonian of the quantum double model.

Our approach relies on and extends the diagrammatic notation of~\cite{kuperberg}. Such notation is an efficient representation of the type of 3D tensor networks describing the transfer matrix and other operators relevant to this paper. It is also very useful for finding the ground states of the quantum double models. We exhibit two ground states, one of which coincides with the one found in~\cite{Aguado} and the other is different from the first when on a surface with non-trivial topology but is the same as the first one on the 2-sphere. This ground state was written in \cite{AG2}. Their representation in terms of a tensor network comes out naturally from the formalism. The ground state degeneracy is shown to be equal to the Kuperberg 3-manifold invariant of $\Sigma \times S^1$ divided by $\textrm{dim}(\mathcal{A})$, for any involutory Hopf algebra.

We organize the paper as follows: we start section \ref{sec-guagetheories} with a brief review on gauge theories by building its partition function using the same notation as in the constraining of Kuperberg's invariant. In order to get the transfer matrix from the partition function a diagrammatic presentation of a tensor network, called Kuperberg diagrams is defined in section \ref{sec-guagetheories}. The Kuperberg invariant is made of two systems of curves such that their weights are built out of the structure constants of an involutory Hopf algebra. In section \ref{sec-lattices} we get the transfer matrix as well as the plaquette and vertex operators of the quantum double models, by using a set of properties of these systems of curves, also described in section \ref{sec-lattices}. The method to obtain the Hamitlonian from the transfer matrix is discussed in section \ref{sec-hamiltonian}. The section \ref{sec-groundstate} is dedicated to the study of the ground states of such models. In this section we explicitly exhibit two ground states and also get the ground state degeneracy for any involutory Hopf algebra. In section \ref{sec-beyond} we write down other models we can obtain using the formalism developed in this paper which we analyze in detail in~\cite{p2}. We close this work with some remarks in section \ref{sec-discussion}.
 
\section{The 3D model and its relation to Lattice Gauge Theories}
\label{sec-guagetheories}
~

The partition function of the 3D model we are about to construct is a generalization of the invariant defined in \cite{CFS,kuperberg}. Lattice gauge theories are particular cases of this more general 3D field theories. We build these theories using the same set of data which are used in \cite{CFS,kuperberg}, namely a lattice discretization $\mathcal{L}(M)$ of a 3D manifold $M$ and the structure constants of an involutory Hopf algebra $\mathcal{A}$. 

The first step is to define and describe the discretization. We begin by introducing a cubic lattice of some 3-manifold $M$ without boundary. The first step will be to encode the 3D lattice structure into what is called the Heegaard diagram. One advantage of using a Heegaard diagram is that they are two-dimensional. Using 2D diagrams instead of 3D lattices helps in making some of the computations more transparent. More importantly, Heegaard diagrams are more flexible and allow for manipulations that are hard to describe or juts do not make sense when the model is written in a conventional lattice. This is specially true when one investigate topological invariance.  

The second step is to introduce a set of weights to build the partition function $Z$. They are given by a set of tensors with covariant and contravariant indices and $Z$ is a scalar defined by a certain tensor network. It has become standard in physics to use a graphical notation to describe tensor networks \cite{Treview}. In this paper we will use the diagrams introduced by Kuperberg~\cite{kuperberg,Treview}. 

\subsection{Diagrams and Lattices}
~

We recall here how one can canonicaly associate a 2D diagram ${\cal D}$ to a 3D lattice discretization ${\cal L}$. All information contained in ${\cal L}$ will be encoded in $\mathcal{D}$. For convenience we will use cubic lattices but the procedure can be applied for any lattice such as triangulations by tetrahedra. 

Consider $M$ to be a $3-$manifold without a boundary $\left( \partial M=\emptyset\right)$. Let $\mathcal{L}$ be a lattice discretization of $M$. This lattice is made of cubes glued together by their faces, as illustrated in figure \ref{colagemcubos}.
\begin{figure}[h!]
	\begin{center}
		\includegraphics[scale=1]{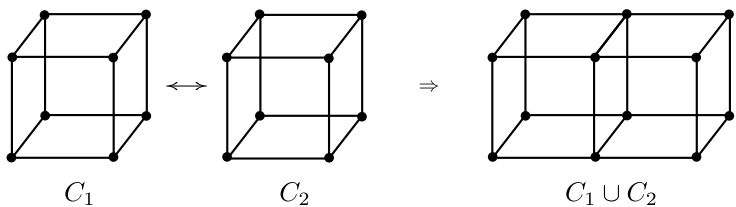}
	\caption{An example of two cubes glued together by their faces.}
	\label{colagemcubos}
	\end{center}
\end{figure}

The diagram which represents the lattice $\mathcal{L}$ comprise of two sets of curves drawn on a 2D surface $H$. One type of curve is represented by red curves and the other by blue curves.\begin{figure}[htb!]
\centering
\subfigure[The $1-$skeleton associated with one cube.]{
\includegraphics[scale=1]{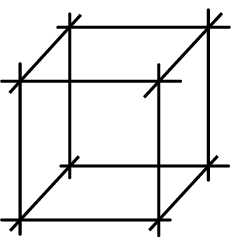}\label{1skeleton-a}
}
\hspace{3cm}
\subfigure[The tubular neighborhood of the $1-$skeleton of figure \ref{1skeleton-a}.]{
\includegraphics[scale=1]{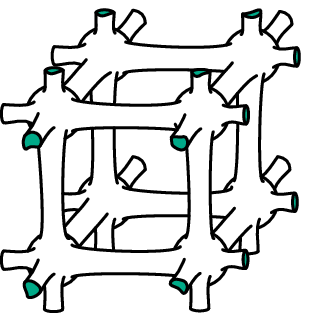}\label{1skeleton-b}
}
\caption{Defining the space $H$ where the systems of curves lie.}
\label{1skeleton}
\end{figure} The surface $H$ is given as follows. Consider the 1-skeleton of $\mathcal{L}$, in other words, just the skeleton composed of links and vertexes  of $\mathcal{L}$ and then consider its tubular neighborhood $H$, as shown in figure \ref{1skeleton}.

The space $H$ is compose of balls (one for each vertex) and cylinders (one for each link) glued accordingly. Now we draw one red curve on each cylinder of $H$ (figure \ref{systemofcurves-a}) to represent the links. We also draw a blue curve for each plaquette (face) of ${\cal L}$ as indicated in figure \ref{heegaarddiagrams-a}).
\begin{figure}[htb!]
\centering
\subfigure[The red curves associated with the links.]
{
\includegraphics[scale=1]{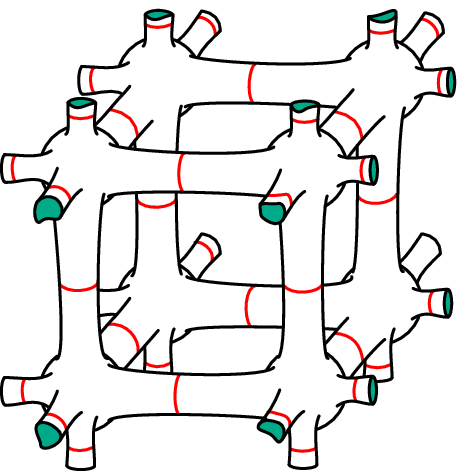} \label{systemofcurves-a}
}
\hspace{2cm}
\subfigure[The diagram associated with a cube of $\mathcal{L}$.]
{
\includegraphics[scale=1]{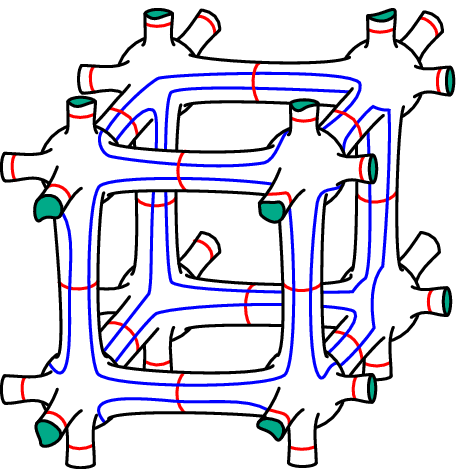} \label{heegaarddiagrams-a}
}
\caption{Systems of curves on the surface $H$.}
\label{systemscurves}
\end{figure}
As one can see, the read curves are placed around the links and the blue curves follows the boundary of each plaquette.

Let $\mathcal{D}$ be the entire diagram made of a lot of cells like the one show in figure \ref{systemscurves} glued together. Notice that a face with four sides corresponds to a blue curve crossed by four read curves as indicated in figure \ref{bluecurvesquare-a}. In a similar way, a link where four faces are joined is depicted by a red curve that crosses four blue curves as shown in \ref{bluecurvesquare-b}.
\begin{figure}[htb!]
\centering
\subfigure[A blue curve associated with a plaquette of $\mathcal{L}$.]
{
\includegraphics[scale=1]{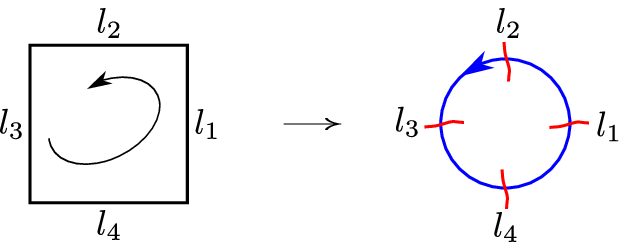} \label{bluecurvesquare-a}
}
\hspace{2cm}
\subfigure[A red curve associated with a link of $\mathcal{L}$.]
{
\includegraphics[scale=1]{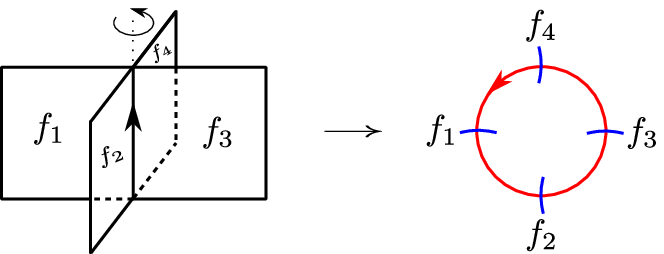} \label{bluecurvesquare-b}
}
\caption{Definition of the two system of curves.}
\label{bluecurvesquare}
\end{figure}

A final remark about orientations is in order. The models to be defined here are  generalizations of usual lattice gauge theories. In order to write down the partition function we need to choose an orientation for each face and each link of ${\cal L}$. However the partition function does not depend on the choice of orientation, in the same as it also happens for usual lattice gauge theories. 
The orientation of a face is encoded by an orientation of the corresponding blue curve. In the same way, we can read the orientation of a link by looking at the corresponding red curve. The rule should be clear from figure \ref{bluecurvesquare}.

\subsection{Partition Function}
~

The way we  will construct the partition function $Z$ is based on the procedure used in \cite{CFS,kuperberg}. We associate weights given by tensors to each curve of $\mathcal{D}$ and define $Z$ as a scalar costructed out of these tensors. The data we need to define the weights is given by the structure constants of an involutory Hopf algebra. Thus let $\mathcal{A}$ be an involutory Hopf algebra, in other words, $\mathcal{A}=\langle m,\eta, \Delta, \epsilon,S \rangle$, where $m$ and $\Delta$ are the multiplication and co-multiplication maps, $\eta$ and $\epsilon$ are the unit and co-unit of $\mathcal{A}$ and $S$ is the antipode such that $S^2 = \mathbf{1}$. In appendix \ref{ap-hopfalgebras} one can find a definition of a Hopf algebra as well as the proof of the identities on Hopf algebras relevant for this work. In this paper, we consider only square lattices and so all the blue and red curves of ${\cal D}$ have exactly four crossings.

Let $\{ \phi_i \}$ be a basis of the algebra $\mathcal{A}$. For each blue curve (plaquette of $\mathcal{L}$) we associate a covariant tensor $M_{l_1 l_2 l_3 l_4}$, like the one shown in figure \ref{bluecurvesquare-a}, here $l_1$, $l_2$, $l_3$ and $l_4$ represent the links which cross this curve. The tensor $M_{l_1 l_2 l_3 l_4}$ is defined by
\begin{equation}
M_{l_1 l_2 l_3 l_4} = \textrm{tr}\left(z~\phi_{l_1}\phi_{l_2}\phi_{l_3}\phi_{l_4}\right)
\label{mtensordefinition}
\end{equation}
where $\textrm{tr}$ means the trace in the regular representation and $z$ is some element belonging to the center of the algebra (see figure \ref{blue-weight}). \begin{figure}[h!]
\centering
\subfigure[The weight associated with a blue curve.]{
\includegraphics[scale=1]{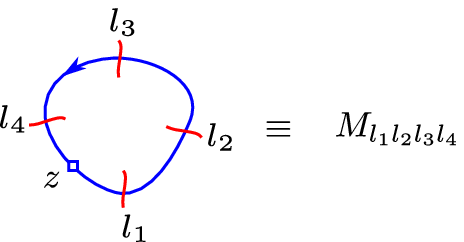} \label{blue-weight}
}
\hspace{2.5cm}
\subfigure[The weight associated with a red curve.]{
\includegraphics[scale=1]{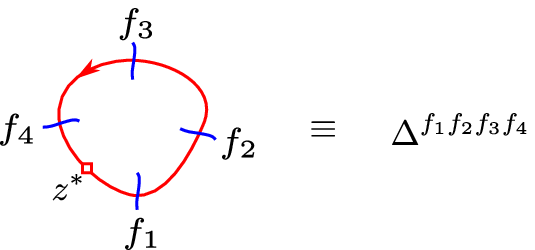} \label{red-weight}
}
\caption{The weights associated with the system of curves.}
\label{blue-red-weight}
\end{figure}  Analogously we associate a contra-variant  tensor $\Delta^{f_1 f_2 f_3 f_4}$ for each red curve of $\mathcal{D}$ (figure \ref{red-weight}), where $f_1$, $f_2$, $f_3$ and $f_4$ represent the faces which share this link. The $\Delta^{f_1 f_2 f_3 f_4}$ tensor is defined by 
\begin{equation}
\Delta^{f_1 f_2 f_3 f_4} = \textrm{cotr}\left(z^*~\varphi_{f_1}\varphi_{f_2}\varphi_{f_3}\varphi_{f_4}\right)\;,
\label{deltatensordefinition}
\end{equation}
where $\{\varphi^j\}$ is a basis of $\mathcal{A}^*$ such that $\varphi^j\left( \phi_i\right)=\delta(i,j)$ and $z^*$ is an element which belongs to the co-center of $\mathcal{A}$. Note that we have placed a small rhombus on the curves. They carry labels  $z$ or $z^*$ as in figures \ref{blue-weight} and \ref{red-weight} and indicate the particular elements entering in the definition of  weights (\ref{mtensordefinition}) and (\ref{deltatensordefinition}). We say that that curves are colored by an element $z$ or $z^*$. In the special case where $z$ is the unit of the algebra (or $z^*$ is the co-unit of the algebra), we represent this as a single curve without these little rhombus, see figure \ref{hole}. When a curve is colored with unit (or co-unit) we refer to them as being trivially colored. The weighs defined in \cite{CFS,kuperberg} are trivially colored.
\begin{figure}[h!]
\centering
\subfigure[]{
\includegraphics[scale=1]{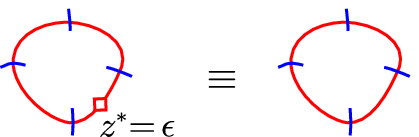} \label{hole-a}
}
\hspace{2.5cm}
\subfigure[]{
\includegraphics[scale=1]{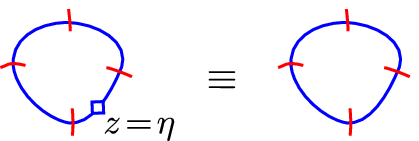} \label{hole-b}
}
\caption{In \ref{hole-a} a trivial red curve, colored by co-unit of $\mathcal{A}$, and in \ref{hole-b} a trivial blue curve, colored by unit of $\mathcal{A}$.}
\label{hole}
\end{figure}

The partition function we are building is made of contractions between these two kinds of tensors, but before we go ahead we have to take into account the orientation on each curve. As stated before, curve orientation encode lattice orientation. We also need to fix a orientation for the surface where the curves are lying. By convention let us take the normal vector to the surface pointing out. For each crossing between a blue and a red curve we have to contract the corresponding index of each curve. This contraction can be direct or indirect according to the convention shown in figure \ref{orientacaocurvas}. The vector $\hat{n}$ is the normal vector to the surface, so when $\vec{s}_b \times \vec{s}_r$ is parallel to $\hat{n}$ we contract the curves as shown in figure \ref{orientacaocurvas-b}, otherwise we use the antipode map $S$ to make the contraction, as shown in figure \ref{orientacaocurvas-a}.
\begin{figure}[h!]
\centering
\subfigure[]{
\includegraphics[scale=1]{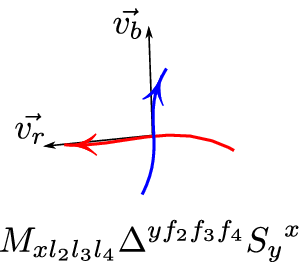} \label{orientacaocurvas-a}
}
\hspace{2.5cm}
\subfigure[]{
\includegraphics[scale=1]{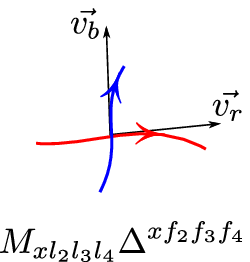} \label{orientacaocurvas-b}
}
\caption{The contraction rule for the tensors $M_{l_1 l_2 l_3 l_4}$ and $\Delta^{f_1 f_2 f_3 f_4}$.}
\label{orientacaocurvas}
\end{figure}
The contraction of the tensors associated with all the curves gives us the partition function 
\begin{equation}
Z\left(z,z^* \right) = \prod_{\hbox{red  curves}}\Delta^{f_1 f_2 f_3 f_4} \prod_{\hbox{blue curves}} M_{l_1 l_2 l_3 l_4}\prod_{\hbox{orient.}}{S_x}^y\;.
\label{invariante}
\end{equation}
Notice that since all indices are contracted $Z$ is a scalar. 
Such a scalar can be viewed as given by a tensor network constructed out of $M_{l_1 l_2 l_3 l_4}$, $\Delta^{f_1 f_2 f_3 f_4}$ and ${S_x}^y$. The pattern of contractions defining the network is determined by the lattice or, equivalently, by the crossings of blue and red curves of ${\cal D}$. 
We recal once more that the models defined in \cite{CFS,kuperberg} correspond to $z=\eta$ and $z^*=\epsilon$.

\subsection{Gauge Theories}
~

Let us briefly discuss how the partition function of lattice gauge theories are obtained form (\ref{invariante}). In this case, fields living on the links are elements of the gauge group $G$. Therefore $\mathcal{A}$ has to be the group algebra $\mathbb{C}(G)$. For simplicity, let us assume that $G$ is finite dimensional and let  $\{\phi_g\}, g\in G$ be the basis of $\mathbb{C}(G)$.
The weight associated to the faces is the Boltzmann factor
\begin{equation}
{^GM}_{{g_1} g_2 g_3 g_4}=e^{-\beta \textrm{tr}\left(g_1 g_2 g_3 g_4 \right)}
\label{Mgauge}
\end{equation}
where $g_1$, $g_2$, $g_3$ and $g_4$ are the variables which live on the boundary of some face. Note that the tensor $M_{g_1 g_2 g_3 g_4}$ is invariant under cyclic permutation of their indices. For pure gauge theories there is no weight associated to the links, but faces shared by the same link have to agree on the same variable, which means that the tensor $\Delta^{g_1 g_2 g_3 g_4}$ has to be
\begin{equation}
{^G\Delta}^{g_1 g_2 g_3 g_4}= \delta(g_1,g_2)\delta(g_1,g_3)\delta(g_1,g_4)
\label{Dgauge}
\end{equation}
where $g_1$, $g_2$, $g_3$ and $g_4$ are the variables at faces which share the same link. If we want to define a gauge theory on the diagrams, instead of associating a weight to the faces, we have to associate a weight to the blue curves. But the way to do that is straightforward we just associate the same weight to the blue curves. The next step is to choose $z$ and $z^*$ so as to reproduce the tensors ${^GM}_{g_1 g_2 g_3 g_4}$ and ${^G\Delta}^{g_1 g_2 g_3 g_4}$.

Consider as an example $G=\mathbb{Z}_2$. In this case the group algebra has only two elements in its basis $\{ \phi_0, \phi_1\}$. If we want to reproduce the tensors (\ref{Mgauge}) and (\ref{Dgauge}) we just need to make the following choices for $z$ and $z^*$ \cite{MJBV}.
$$z=\frac{1}{2}e^{\beta}~\phi_0+\frac{1}{2}e^{-\beta}~\phi_1 \;\;\;\; \hbox{and} \;\;\;\; z^*=\epsilon$$
and then the partition function becomes
$$^GZ(\beta)= \prod_f M(f)$$
where $M(f)={^GM}_{g_1 g_2 g_3 g_4}$.

In the special case where $z=\eta=\phi_0$ and $z^*=\epsilon$, the function (\ref{invariante}) is proportional to the topological invariant defined in \cite{CFS, kuperberg}. In other words, the gauge theories, in the limit $\beta \rightarrow - \infty$, become topological \cite{MJBV}.

\subsection{Kuperberg's Diagrams}
~

So far we have considered only a manifold without boundaries. In the following we start to look at field theories in $(2+1)D$. That means that our manifold is of a particular kind $\Sigma \times \left[0,~1 \right]$ which is a 3-manifold with boundary. The way we will do that is by defining the transfer matrix such that its trace is the partition function of the system. 

It will be convenient to regard $Z$ as a scalar given by a tensor network. As such, a graphical notation turns out to be useful. 
In this section we introduce a notation which will be well adapted to our purposes. This notation is essentially the same as the one adopted for example in \cite{Treview}. The main difference being the distinction between covariant and contravariant indices. Since it has been introduced in Kuperberg's work \cite{kuperberg} we call this {\it Kuperberg's notation}. In particular, the tensors given by the structure constants of a Hopf algebra can also be represented in this fashion. It turns out that this diagramatic notation simplifies some of the algebraic manipulations we need to perform.

\subsection{Kuperberg's Notation}
~

A Kuperberg diagram is defined in the following way: consider a generic tensor ${T_{a_1 a_2 \cdots a_n}}^{b_1 b_2 \cdots b_m}$  which belongs to the space $V\otimes V \otimes \cdots \otimes V \otimes V^{*} \otimes V^{*} \otimes \cdots \otimes V^{*}$. We associate a diagram $T$ to this tensor. We represent each covariant index by an arrow coming into the diagram $T$, and each contravariant index by an arrow going out of the diagram $T$. By convention, the arrows coming in are enumerated counter clockwise and the arrows going out are enumerated clockwise. See figure \ref{grafosecontracoes-a}.

\begin{figure}[htb!]
\centering
\subfigure[]{
\includegraphics[scale=1]{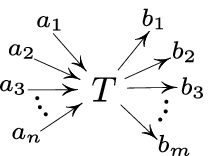} \label{grafosecontracoes-a}
}
\hspace{3cm}
\subfigure[]{
\includegraphics[scale=1]{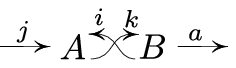} \label{grafosecontracoes-b}
}
\hspace{3cm}
\subfigure[]{
\includegraphics[scale=1]{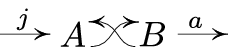} \label{grafosecontracoes-c}
}
\caption{{\bf (a)} The Kuperberg diagram associated with the tensor ${T_{a_1 a_2 \cdots a_n}}^{b_1 b_2 \cdots b_m}$. {\bf (b)} The contraction rule between two tensors.}
\label{grafosecontracoes}
\end{figure}

The contraction rule in terms of Kuperberg diagrams involves connecting the arrows which are contracted. Consider the following contraction: ${A_{ij}}^k {B_k}^{ai}$, the corresponding Kuperberg diagram is the one shown in figure \ref{grafosecontracoes-b}. Once summed, we do not need to write them down on the diagram (see figure \ref{grafosecontracoes-c}).

In the case when $T$ is a linear transformation, from a vector space $V$ on itself, we write $T$ as a tensor with one covariant and one contravariant index ${T_a}^b$, therefore we represent this as a diagram with one arrow coming in and one arrow going out, as shown in figure \ref{grafosmapaslineares-a}. The identity map is the one shown in figure \ref{grafosmapaslineares-b}.

\begin{figure}[htb!]
\centering
\subfigure[]{
\includegraphics[scale=1]{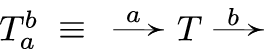} \label{grafosmapaslineares-a}
}
\hspace{4cm}
\subfigure[]{
\includegraphics[scale=1]{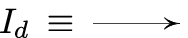} \label{grafosmapaslineares-b}
}
\caption{{\bf (a)} The Kuperberg diagram of a linear transformation. {\bf (b)} The Kuperberg diagram of the identity map of $V$. }
\label{grafosmapaslineares}
\end{figure}

There is a very important quantity which is the trace of an operator, the Kuperberg diagram associated to this quantity is a diagram which has an arrow that goes out and comes in to the same diagram. For a linear operator $A_i^j$ the trace is given by
$$\textrm{tr}\left(A \right)={A_i}^i = \sum_i {A_i}^i,$$
it means that the indices of ${A_i}^j$ are being contracted. Therefore the Kuperberg diagram associated with the trace of a tensor is the one in figure \ref{tracooperador-a}. We represent the trace of the identity map as a single closed curve, as shown in figure \ref{tracooperador-b}.

\begin{figure}[htb!]
\centering
\subfigure[]{
\includegraphics[scale=1]{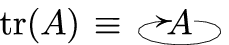} \label{tracooperador-a}
}
\hspace{4cm}
\subfigure[]{
\includegraphics[scale=1]{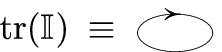} \label{tracooperador-b}
}
\caption{{\bf (a)} The Kuperberg diagram of the trace of an arbitrary linear map. {\bf (b)} The trace of the identity map of $V$.}
\label{tracooperador}
\end{figure}

We now write all the weights associated with the curves in terms of Kuperberg's notation.

\subsection{Weights in Kuperberg's Notation}
~

Consider the blue curve in the figure \ref{plaquette-weight-a}. The weight associated to this curve is the tensor $M_{l_1 l_2 l_3 l_4}$ defined in equation (\ref{mtensordefinition}) and since it is the trace in the regular representation it can be written in terms of the structure constants of the algebra (see proposition (\ref{prop-trace}) in appendix \ref{ap-hopfalgebras}). Therefore its weight is the one written in figure \ref{plaquette-weight-b}, where we can still use the associativity (appendix \ref{ap-hopfalgebras}, figure \ref{condicaotensormultiplicacao}) of the algebra in order to write this tensor in a more symmetric way, as in figure \ref{plaquette-weight-c}.

\begin{figure}[h!]
\centering
\subfigure[]{
\includegraphics[scale=1]{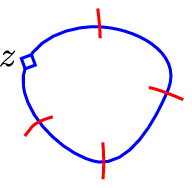} \label{plaquette-weight-a}
}
\hspace{.7cm}
\subfigure[]{
\includegraphics[scale=1]{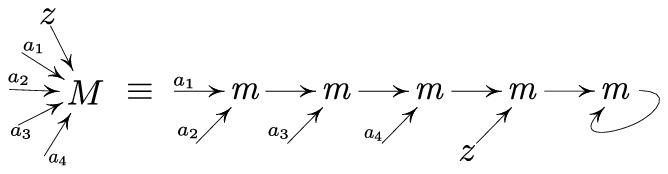} \label{plaquette-weight-b}
}
\hspace{.7cm}
\subfigure[]{
\includegraphics[scale=1]{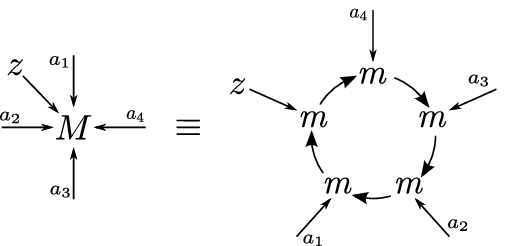} \label{plaquette-weight-c}
}
\caption{In {\bf (a)}, the little rhombus with the element $z$ hanging on it means that this curve is colored by the element $z$. In {\bf (b)} and {\bf (c)} the tensor associated with a blue curve with four crossings.}
\label{plaquette-weight}
\end{figure}

Also the weight associated to the links can be written in terms of Kuperberg's diagrams. The curve drawn in figure \ref{link-weight-a} represents one red curve colored by an element $z^*$ and its weight is the one in figure \ref{link-weight-b}, where again we have used proposition (\ref{prop-trace}) to write it in terms of the structure constants of the co-algebra. Using co-associativity we can also write this in a more symmetric way, as shown in figure \ref{link-weight-c}.

\begin{figure}[h!]
\centering
\subfigure[]{
\includegraphics[scale=1]{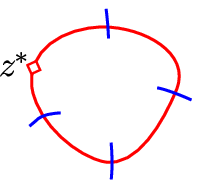} \label{link-weight-a}
}
\hspace{.7cm}
\subfigure[]{
\includegraphics[scale=1]{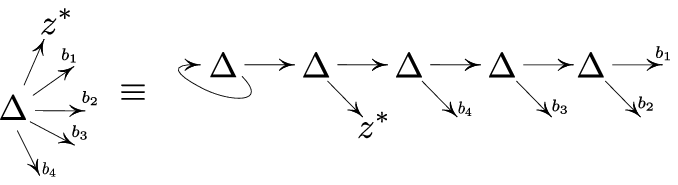} \label{link-weight-b}
}
\hspace{.7cm}
\subfigure[]{
\includegraphics[scale=1]{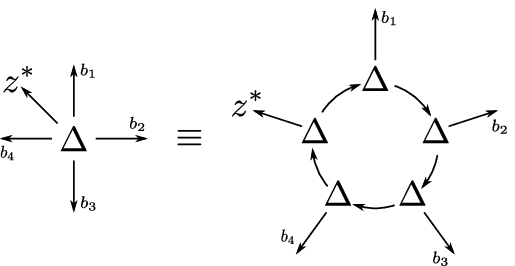} \label{link-weight-c}
}
\caption{In {\bf (a)}, the little rhombus with the element $z^*$ hanging on it means that this curve is colored by the element $z^*$. In {\bf (b)} and {\bf (c)} the tensor associated with a red curve with four crossings.}
\label{link-weight}
\end{figure}

\section{Lattices with boundaries and the Transfer Matrix}
\label{sec-lattices}
~

Consider a manifold of the form $\Sigma\times S^1$ where the compact ``time" direction has being discretized into $N$ steps. From the partition function we can derive a Hamiltonian operator $H$ such that 
$$
Z(\Sigma\times S^1)=\textrm{Tr}\left(\left[e^{-H \delta t}\right]^N \right),
$$
where $U=e^{\left(-\delta t ~ H\right)}$ is the transfer matrix. This can be thought of as being the time evolution operator in $(2+1)D$. For the purpose of obtaining the Hamiltonian, it is enough to consider $N=1$. 
In Kitaev's model there is one quantum state ($\left|+\right>$ or $\left|-\right>$ in the case of $\mathbb{C}(\mathbb{Z}_2)$.) living on each link of the lattice, such states belongs to a Hilbert space $H$. For the quantum double model, a basis for $H$ is $\{|g\rangle\},g\in G$. In any case,
there is one vector associated to each link. The Hilbert space for the entire lattice is $\mathcal{H}=\underbrace{H \otimes H \otimes \cdots \otimes H}_{n_l \hbox{ times}}$, where $n_l$ is the number of links for the lattice discretization of $\Sigma$. The transfer matrix is a map $U:\mathcal{H}\rightarrow \mathcal{H}$ and therefore can be viewed as a tensor in $\mathcal{H}\otimes\mathcal{H}^*$. In terms of Kuperberg's notation it will be represented by a diagram with $n_l$ arrows coming in and $n_l$ arrows going out. 

The processes of obtaining $U$ can be visualized as follows. The partition function for $N=1$ is given by some tensor network represented by $Z(\Sigma\times S^1)$ on figure \ref{desenho1}. Each link of $\Sigma$ will contribute with a contraction as indicated in the figure. The operator $U$ is the splitting of such network along the links of $\Sigma$. As explained ahead in this section, the network $Z$  can be encoded in a diagram ${\cal D}$ of blue and read curves. It will be  convenient to also encode $U$ with a similar diagram of curves. That can be done provided we improve the diagram in order to include more tensors other then just $M_{l_1 l_2 l_3 l_4}$ and $\Delta^{f_1 f_2 f_3 f_4}$.
\begin{figure}[h!]
\centering
\subfigure[]{
\includegraphics[scale=1]{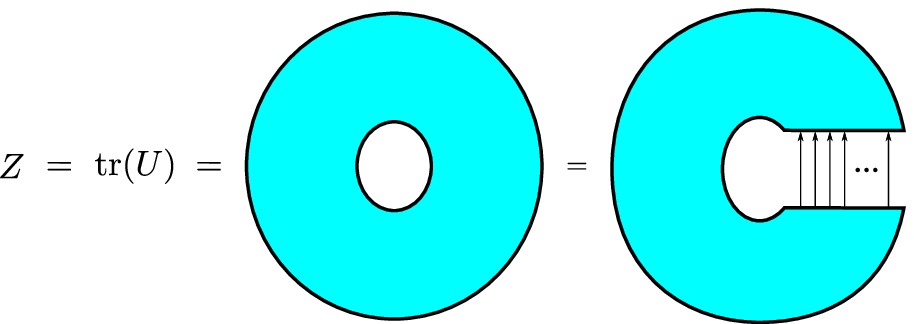} \label{desenho1-a}
}
\hspace{1.5cm}
\subfigure[]{
\includegraphics[scale=1]{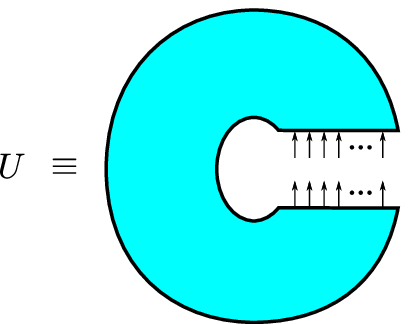} \label{desenho1-b}
}
\caption{The partition function as the transfer matrix trace.}
\label{desenho1}
\end{figure}

We already know that each red curve cross exactly four blue curves. Let us represent these crossings by blue dots on the red curve, as shown in figure \ref{redcurves-dots-a}. The associated weight is shown in figure \ref{redcurves-dots-b}. Thus the weight associated to the red curve has one free arrow going out for each blue dot on it. In the same way we can build a red curve with blue and red dots, where blue dots mean arrows going out and red dots mean arrows coming in, as illustrated in figure \ref{redcurves-dots-c}. The tensor in figure \ref{redcurves-dots-c} is the co-multiplication tensor. Here the orientation of the curve is very important, the $\Delta$ tensor has indices ordered clockwise starting from the arrow that comes in.
\begin{figure}[h!]
\centering
\subfigure[]{
\includegraphics[scale=1]{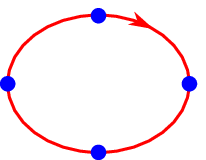} \label{redcurves-dots-a}
}
\hspace{2.5cm}
\subfigure[]{
\includegraphics[scale=1]{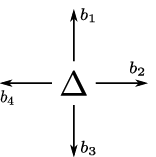} \label{redcurves-dots-b}
}
\hspace{2.5cm}
\subfigure[]{
\includegraphics[scale=1]{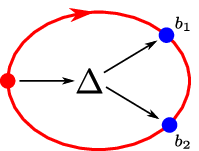} \label{redcurves-dots-c}
}
\caption{The blue dots mean arrows going out of the tensor $\Delta$ while the red one means one arrow coming in.}
\label{redcurves-dots}
\end{figure}
In this notation we can combine curves of the same color, just contracting them by dots of different colors. For example we can combine the two red curves of the figure \ref{redcurves-dots-a} and \ref{redcurves-dots-c}, as shown in \ref{redcurves-combination-a} and \ref{redcurves-combination-b}.
\begin{figure}[h!]
\centering
\subfigure[]{
\includegraphics[scale=1]{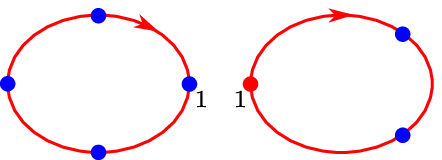} \label{redcurves-combination-a}
}
\hspace{1cm}
\subfigure[]{
\includegraphics[scale=1]{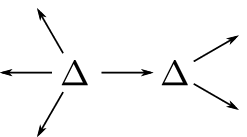} \label{redcurves-combination2-a}
}
\hspace{1cm}
\subfigure[]{
\includegraphics[scale=1]{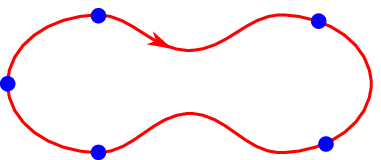} \label{redcurves-combination-b}
}
\caption{Joining two red curves by two dots with different colors. If the picture is read from left to write it shows the splitting of a red curve.}
\label{redcurves-combination}
\end{figure}
We can also read figure \ref{redcurves-combination} from left to right. That corresponds to a splitting of a loop into a pair of curves. This spliting will be used in order to factorize the transfer matrix. Note that in the case of figure \ref{redcurves-combination} both orientations agree, hence the resultant curve is the one shown in figure \ref{redcurves-combination-b}. But, if they do not agree in orientation the contraction has to be done as shown in figure \ref{redcurves-combination-orientation}.
\begin{figure}[h!]
\centering
\subfigure[]{
\includegraphics[scale=1]{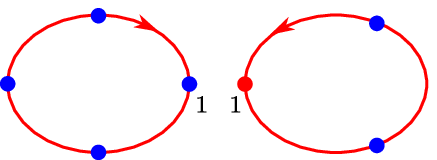} \label{redcurves-combination-orientation-a}
}
\hspace{2.5cm}
\subfigure[]{
\includegraphics[scale=1]{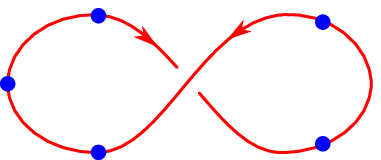} \label{redcurves-combination-orientation-b}
}
\caption{Combination of red curves with different orientations.}
\label{redcurves-combination-orientation}
\end{figure}
In the same way we can put red and blue dots on the blue curves, the meaning of these dots is the same of the ones in the red curves. In figure \ref{bluecurves-dots} we can see the tensor associated with a blue curve with red and blue dots.
\begin{figure}[h!]
\centering
\subfigure[]{
\includegraphics[scale=1]{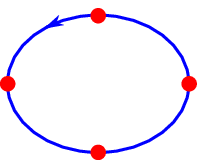} \label{bluecurves-dots-a}
}
\hspace{2.5cm}
\subfigure[]{
\includegraphics[scale=1]{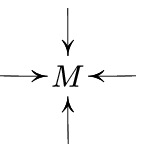} \label{bluecurves-dots-b}
}
\hspace{2.5cm}
\subfigure[]{
\includegraphics[scale=1]{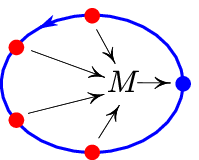} \label{bluecurves-dots-c}
}
\caption{The red dots mean arrows coming into the tensor $M$ while the blue one means one going out.}
\label{bluecurves-dots}
\end{figure}
Note that the indices of the tensor $M$ is ordered counter-clockwise. The rules for composing blue curves is similar to  the one for red curves. On figure \ref{desenho2} we show the gluing and splitting of blue curves with the same orinetation. 
\begin{figure}[h!]
	\begin{center}
		\includegraphics[scale=1.0]{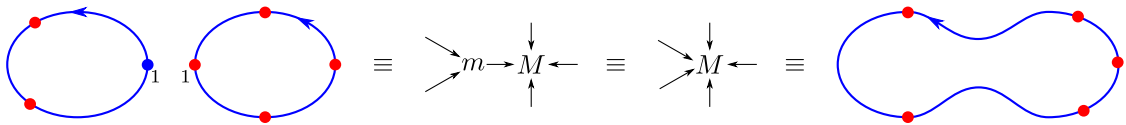}
	\caption{Joining and splitting of blue curves with the same orientations.}
	\label{desenho2}
	\end{center}
\end{figure}

The partition function is made of a lot of cells, like the one shown in figure \ref{heegaarddiagrams-a}, glued one besides the other. Note that faces and links are either timelike or spacelike, so let us call $z_T$ and $z_S$ the elements of the center which color timelike and spacelike faces, respectively. Let us call $z_T^*$ and $z_S^*$ the elements of the co-center which color timelike and spacelike links, respectively. It means that the partition function has different weights associated to timelike faces and spacelike faces, and for the links as well, or in other words, 
$$Z=Z\left(\mathcal{A},z_T,z_S,z_T^*,z_S^* \right).$$

The elements $z$'s and $z^*$'s parametrize the theory. In order to obtain the unperturbed Kitaev's model we have to make a specific choice for these parameters. As we will see later, the following choice is enough to reproduce such a model
\begin{eqnarray}
z_S &=& \frac{1}{n}\left[2(1-\delta_p)~\eta + \delta_p~\lambda\right] \;\;\;\; \hbox{such that} \;\;\; \textrm{tr}(\lambda)=n= \textrm{dim}(\mathcal{A})\;, \label{zs}\\
z_T &=& \eta\; , \nonumber \\
z_S^* &=& \epsilon \; , \nonumber \\ 
z_T^* &=& \frac{1}{n}\left[2(1-\delta_s)~\epsilon + \delta_s~\Lambda \right]\;\;\;\; \hbox{such that} \;\;\; \textrm{cotr}(\Lambda)=n\;, \label{zt*}
\end{eqnarray}
where $\eta$ and $\epsilon$ are the unit and co-unit of $\mathcal{A}$, $\Lambda$ and $\lambda$ are the integral and co-integral of $\mathcal{A}$ and $\delta_{s,p}$ are real parameters which belong to the interval $[0,1]$. In other words we are fixing timelike blue curves and spacelike red curves as being colorless. With this choice of parameters we can reproduce quantum double models.

At this point we are ready to use these tools to split the transfer matrix as a product of operators which acts on links, vertices and plaquettes. To do this we define the transfer matrix operator $U$ such that the partition function described above can be written as $tr(U^N)$. 
For that consider one single cell of the diagram for the partition function, shown on figure \ref{onestepevolution-a}. Consider the diagram ${\cal D}$ obtained by repeating this single cell on the $(x,y)$ with the appropriate boundary condition for the surface $\Sigma$.
This is a lattice with a single link on the ``time" direction. There are red dots on the bottom and blue dots on the top. The corresponding tensor has a pair of arrows in and out for each link of $\Sigma$. 
That is precisely the transfer matrix $U$. 
It is straightforward to see that the partition function for $\Sigma\times S^1$  is obtained by stacking $N$ copies of ${\cal D}$ and connecting the corresponding red dots at the bottom to the blue dots at the top. 

The diagram ${\cal D}$ is not planar but it will be useful to draw its projection on the plane. That can be achieved by selecting a region of the surface as illustrated on figure \ref{onestepevolution-b}. The corresponding projection can be seen on figure  \ref{slice-a}. Some of the curves are not completely contained in this projection and are represented as line segments instead of closed paths. To complete the loops their ends have to be connected. In figure \ref{slice-b} we see a smaller portion of figure \ref{slice-a} where the loops have been completed. The the pattern we see on figure \ref{slice-a} correspond to a plaquette of $\Sigma$. It is clear that we only need to analyze the portion contained in figure \ref{slice-b}. 
\begin{figure}[h!]
\centering
\subfigure[]{
\includegraphics[scale=1]{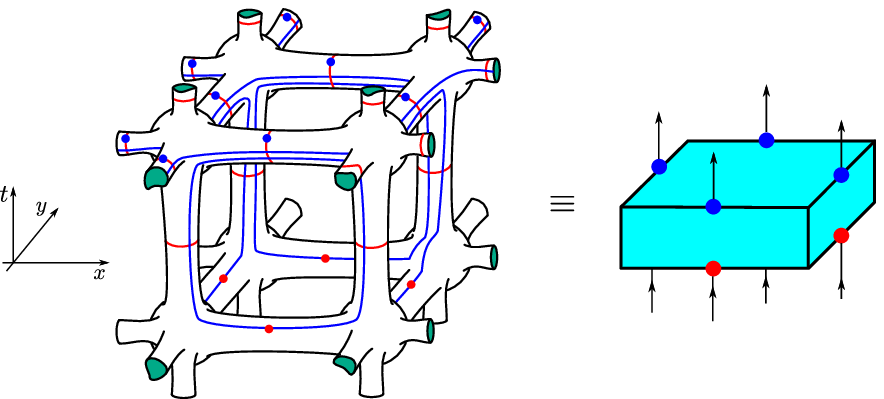} \label{onestepevolution-a}
}
\hspace{1.5cm}
\subfigure[]{
\includegraphics[scale=1]{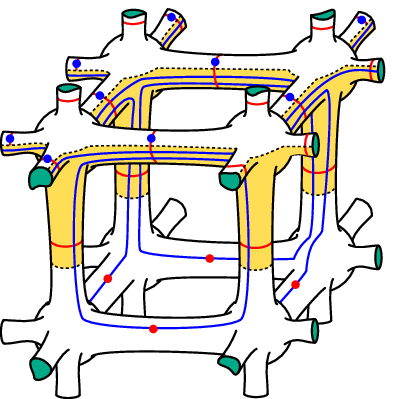} \label{onestepevolution-b}
}
\caption{In {\bf (a)} the transfer matrix $U(t_0,t_1)$. The corresponding tensor has arrows coming going in corresponding to blue and red dots. In {\bf (b)} we select a portion of the surface enough to contain all information in the diagram.}
\label{onestepevolution}
\end{figure}

\begin{figure}[h!]
\centering
\subfigure[]{
\includegraphics[scale=1]{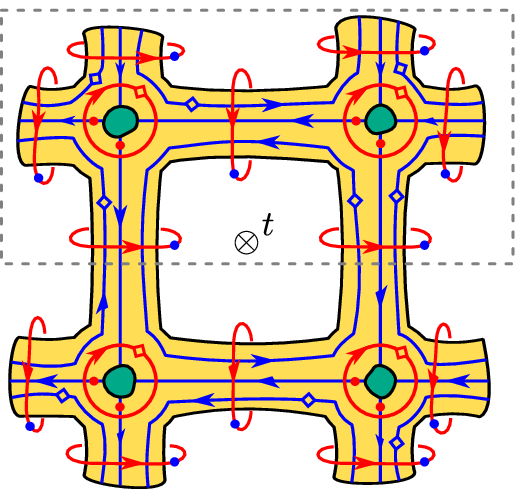} \label{slice-a}
}
\hspace{2cm}
\subfigure[]{
\includegraphics[scale=1]{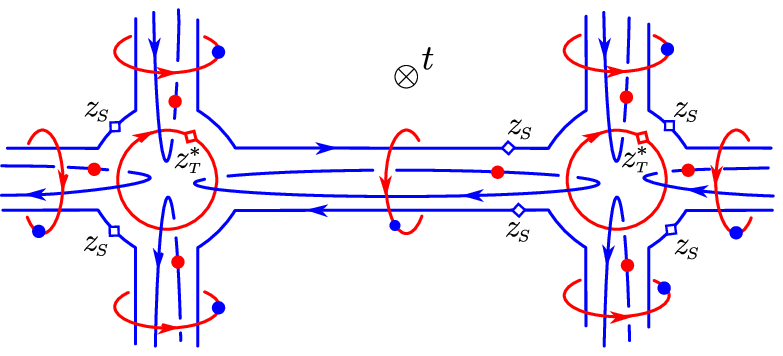} \label{slice-b}
}
\caption{A projection (bottom view) of the shaded region of figure \ref{onestepevolution-a} is in {\bf (a)}. Some of curves are cuted and their ends have to be identified. Figure {\bf (b)} shows a detailed view of the same diagram.}
\label{slice}
\end{figure}

In the following we will see how we can write the transfer matrix as a product of operators which acts on the plaquettes, vertices and links. This factorization is achieved by repeating the splitting of curves described by figures \ref{redcurves-combination} and \ref{desenho2}. The sequence of figures below gives us a prescription of how it can be done. Note that in figure \ref{slice-b} we are just looking at one single link of $\Sigma $. But all the modifications we will perform are local and can be repeated for the entire graph. The first step is to  split the blue loop of figure \ref{split1-a} as in figure \ref{split1-b}. The splitting points will be numbered in order to keep track of them.  
\begin{figure}[h!]
\centering
\subfigure[One spacelike link of the transfer matrix.]{
\includegraphics[scale=1]{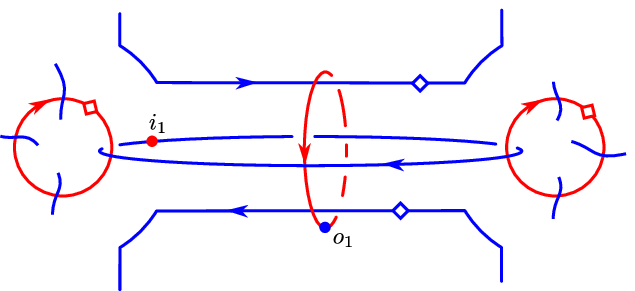} \label{split1-a}
}
\hspace{2cm}
\subfigure[The blue curve in the middle has been broken into two blue curves with red and blue extra dots.]{
\includegraphics[scale=1]{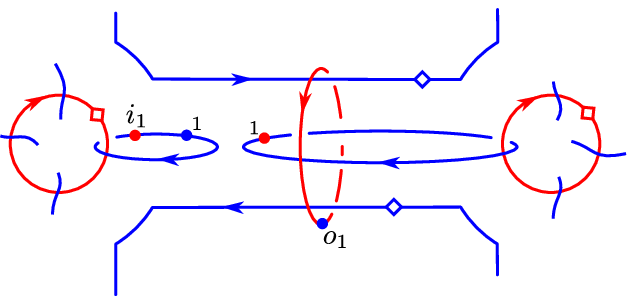} \label{split1-b}
}
\caption{Spliting the diagram to get the operators I.}
\label{split1}
\end{figure}
The next step is to slice the right blue loop of figure \ref{split1-b}. The result is shown in figure \ref{split2-a}. After that we perform the same sequence of splittingsd to the red loop in the middle. The first steps is shown by figure \ref{split2-b}.
\begin{figure}[h!]
\centering
\subfigure[One more split in the middle blue curve.]{
\includegraphics[scale=1]{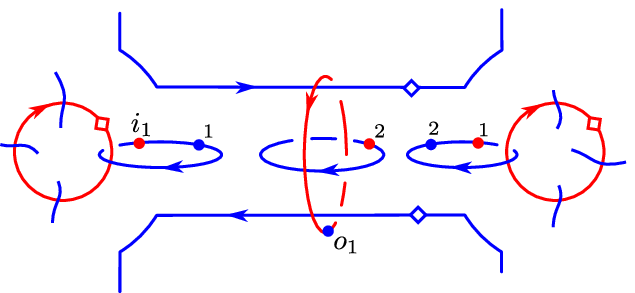} \label{split2-a}
}
\hspace{2cm}
\subfigure[We repeat the same steps in the red curve in the middle.]{
\includegraphics[scale=1]{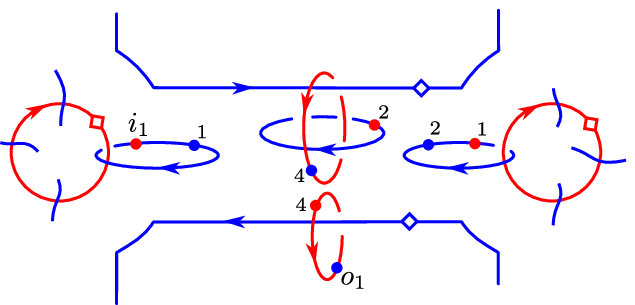} \label{split2-b}
}
\caption{Spliting the diagram to get the operators II.}
\label{split2}
\end{figure}
Finally we get the diagram shown in figure \ref{split3-a}.
\begin{figure}[h!]
	\begin{center}
		\includegraphics[scale=1]{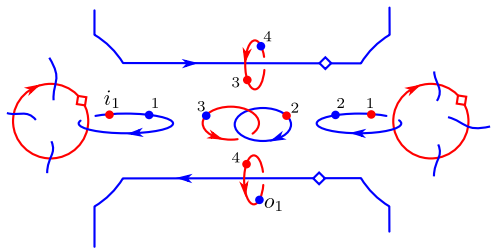}
		\caption{All the curves after being split.}
		\label{split3-a}
	\end{center}
\end{figure}
After applying the same procedure for the entire graph ${\cal D}$
one can see that the transfer matrix is written as a product of the operators shown in figure \ref{operators-diagram}. 
\begin{figure}[h!]
\centering
\subfigure[This operator is called {\it link operator} $R_l$. It acts on a link $l$ of the lattice.]{
\includegraphics[scale=1]{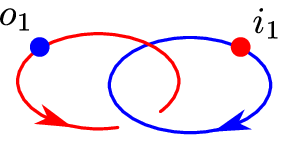} \label{linkoperator-diagram}
}
\hspace{2.2cm}
\subfigure[This operator calls {\it star operator} $A_s$. It acts on an vertex $s$ of the lattice.]{
\includegraphics[scale=1]{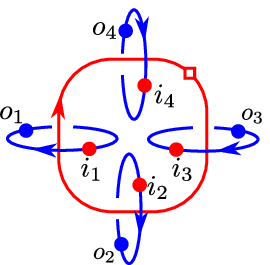} \label{staroperator-diagram}
}
\hspace{2.2cm}
\subfigure[This is the {\it plaquette operator} $B_p$. It acts on a plaquette $p$ of the lattice.]{
\includegraphics[scale=1]{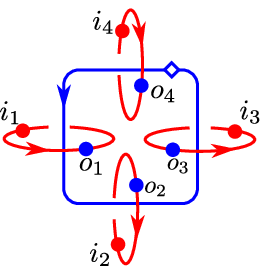} \label{plaquetteoperator-diagram}
}
\caption{The operators which generate the transfer matrix.}
\label{operators-diagram}
\end{figure}

Each vertex, link and plaquette of $\Sigma $ contributes with a star, link and plaquette operator respectively. Any two link operators commute because they act on different links. The same argument goes for the star and plaquette operator. Thus we can write the following commutation relations
\begin{equation}
\left[R_{l_1},R_{l_2} \right] = \left[A_{s_1},A_{s_2} \right] = \left[B_{p_1},B_{p_2} \right] = 0,
\label{trivialcommutation} 
\end{equation}
where $l_i$, $s_i$ and $p_i$ are some link, site and plaquette of the lattice, respectively. 

We now discuss each of these operators.


The link operator $R_l$ shown in figure \ref{linkoperator-diagram} is the simplest one. This operator acts on a link of the lattice, therefore there is one link operator for each link. In terms of Kuperberg diagrams $R_l$ is shown in figure \ref{split3-b}. But due to proposition (\ref{prop-antipodainvoluoria}) appendix A it is easy to see that this operator is proportional to the identity map.
\begin{figure}[h!]
\centering
\subfigure[Link operator written in terms of Kuperberg diagrams.]{
\includegraphics[scale=1]{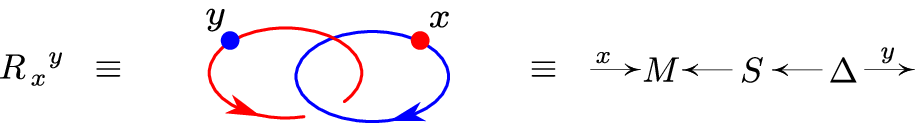} \label{split3-b}
}
\hspace{2cm}
\subfigure[The link operator is proportional to the identity map.]{
\includegraphics[scale=1]{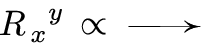} \label{split3-c}
}
\caption{Link Operator.}
\label{split32}
\end{figure}
Thus we can write it down as shown in figure \ref{split3-c}. In figure \ref{split3-a} we can see that this operator plays a role in connecting one plaquette and one star operator which act on the same link, but since this operator is trivial we just connect them directly.

The star operator $A_s(z_T^*)$ given in figure \ref{staroperator-diagram} acts on a vertex $s$ of the lattice. The action changes the states living on the links which share the vertex $s$. The corresponding tensor network is given as a Kuperberg diagram in figure \ref{staroperator-kuperberg}.
\begin{figure}[h!]
	\begin{center}
		\includegraphics[scale=1]{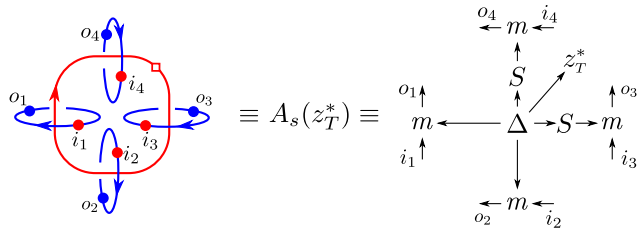}
	\caption{Star operator written in terms a of Kuperberg diagram.}	
	\label{staroperator-kuperberg}
	\end{center}
\end{figure}


The plaquette operator $B_p(z_S)$  represented by figure \ref{plaquetteoperator-diagram} acts on a plaquette $p$ of the lattice. Its Kuperberg diagram is given in figure \ref{plaquetteoperator-kuperberg}.
The action depends on the central element $z_s$. For the case of $\mathcal{A}=\mathbb{C}(G)$ it is a sum of projectors which project onto the different conjugacy classes of $G$ ~\cite{p2}. 
\begin{figure}[h!]
	\begin{center}
		\includegraphics[scale=1]{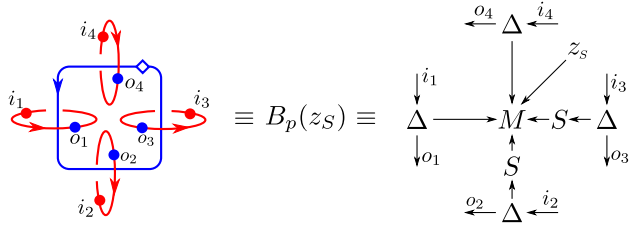}
	\caption{Plaquette operator written in terms of a Kuperberg diagram.}	
	\label{plaquetteoperator-kuperberg}
	\end{center}
\end{figure}

\section{Obtaining the Hamiltonian from the Transfer Matrix}
\label{sec-hamiltonian}
~

In the previous section the partition function was written down as the trace of the transfer matrix $U$. This operator is the product of all plaquette operators and all star operators. It turns out that the star and plaquette operators commute with each other. This fact is very easy to see for the case $\mathcal{A}=\mathbb{C}(\mathbb{Z}_2)$ but becomes more elaborated in general or even for group algebras of non-abelian groups. A proof can be found in appendix \ref{ap-commutation}. The commutation of plaquette and star operators allows us to write the transfer matrix as a product of plaquette and star operators in the most convenient way. Thus consider the transfer matrix
\begin{eqnarray}
U(\mathcal{A},z_S,z_T^*)&=&\prod_p B_p(z_S)\prod_s A_s(z_T^*)\prod_l R_l \nonumber \\
&=&\left(\textrm{dim}(\mathcal{A}) \right)^{n_l}~\prod_p B_p(z_S)\prod_s A_s(z_T^*),
\label{evolutionoperator}
\end{eqnarray}
where $n_l$ is the total number of links.
The Hamiltonian H can now be found from
$$U=e^{- \Delta t~H}$$
by taking the logarithm of $U$.

Consider the plaquette operator $B_p(z_S)$ in figure \ref{plaquetteoperator-kuperberg}, where $z_S$ is the one chosen in (\ref{zs}). Since the tensor $M_{abcd}$ is defined as the trace in the regular representation it is a linear function of $z$, so we can write this operator in the following way
\begin{equation}
B_p(z_S)=1/n\left[2(1-\delta_p)~B_p(\eta)+\delta_p~B_p(\lambda)\right].
\label{op-eq1}
\end{equation}
The operator $B_p(\eta)$ it is colored by the unit element of $\mathcal{A}$. It is represented by a single blue curve without a rhombus and will be denoted by $B_p^0$. The operator $B_p(\lambda)$ is the plaquette operator colored with the co-integral, and it is proportional to the identity map. These operators fulfil property equations
\begin{eqnarray}
\left(B_p^0\right)^k  & = & \left(\textrm{dim}(\mathcal{A})\right)^{k-1} B_p^0, \;\;\;\; k \neq 0 \label{eq1} \\
B_p(\lambda) &=& \underbrace{\textrm{tr}(\lambda)}_{=n}\underbrace{\left(\mathbf{1} \otimes \mathbf{1} \otimes\mathbf{1} \otimes\mathbf{1}\right)}_{\doteq \mathbf{1}}  = n ~ \mathbf{1}.\label{eq2} 
\end{eqnarray}
A proof can be found in appendix \ref{ap-powers}.
Using equation (\ref{eq2}) in (\ref{op-eq1}) the operator $B_p(z_S)$ can be written in the following way
\begin{eqnarray}
B_p(z_S)&=& 2/n~(1-\delta_p)~B_p^0 + \delta_p~\mathbf{1} \nonumber \\
&=& -(1-\delta_p)\left[-2/n~ B_p^0 +\mathbf{1}\right] +(\delta_p +1-\delta_p)~\mathbf{1} \nonumber \\
&=&-(1-\delta_p) B_p(z_0)+\mathbf{1}\;\nonumber
\label{eq3}
\end{eqnarray}
where $n=\textrm{dim}(\mathcal{A})$ and $B_p(z_0)=\mathbf{1}-2/n~ B_p^0$. The operator $B_p(z_0)$ obeys the following property
\begin{equation}
\left(B_p(z_0)\right)^k=\left\{\begin{array}{ccc}
\mathbf{1} & \hbox{if} & k=\hbox{even}\;, \\
B_p(z_0) & \hbox{if} & k=\hbox{odd}\;.
\end{array}
\right.
\end{equation}

Now we are ready to compute the hamiltonian by taking the logarithm of the transfer matrix. For $\Delta t=1$ we obtain
\begin{eqnarray}
-H &=& \ln{\left(n^{n_l}~  \prod_p B_p(z_S)\prod_s A_s(z_T^*)\right)} \nonumber \\
&=&\ln{\left(n^{n_l}\right)} +\sum_p \ln{\left[B_p(z_S)\right]} +\sum_s \ln{\left[A_s(z_T^*)\right]}. \label{eq4}
\end{eqnarray}
Using the Taylor expansion of $\ln(1+x)$ and equation (\ref{eq3}) we can compute the logarithm of the plaquette operator
\begin{eqnarray}
\ln{\left[B_p(z_S)\right]} &=& \ln{\left[\mathbf{1}+(\delta_p -1) B_p(z_0)\right]} \nonumber \\
&=&\sum_{k=1}^{\infty}\frac{(-1)^{k+1}}{k}\left[(\delta_p -1) B_p(z_0) \right]^k \nonumber \\
&=& -\sum_{\hbox{even}}\frac{1}{k}(\delta_p -1)^k + B_p(z_0)~\sum_{\hbox{odd}}\frac{1}{k}(\delta_p -1)^k\nonumber \\
&=& \frac{1}{2} \ln{\delta_p(2-\delta_p)} + \frac{1}{2} \ln\left( \frac{\delta_p}{2-\delta_p}\right)B_p(z_0)\;. \label{eq5}
\end{eqnarray}

All the arguments used for plaquette operator hold for the star operator, since they are constructed in a analogous way.  This is a consequence of the duality for Hopf algebras (between the algebra and co-algebra structures). Therefore we can write
\begin{equation}
\ln{\left[A_s(z_T^*)\right]} =\frac{1}{2} \ln{\delta_s(2-\delta_s)} + \frac{1}{2} \ln\left( \frac{\delta_s}{2-\delta_s}\right)A_s(z_0^*),\;.\label{eq6}
\end{equation}
where $A_s(z_0^*)=\mathbf{1}-2/n~ A_s^0$ and $A_s^0=A_s(\epsilon)$. Finally using equations (\ref{eq5}) and (\ref{eq6}) in (\ref{eq4}) gives us the Hamiltonian
\begin{equation}
H(\gamma_p,\gamma_s)= -~\gamma_p~\sum_p \left(\mathbf{1}-\frac{2}{n}~B_p^0 \right) -~\gamma_s~\sum_s\left(\mathbf{1}-\frac{2}{n}~A_s^0\right) -\gamma ~\mathbf{1} ,
\label{hamiltoniankitaev}
\end{equation}
where the parameters $\gamma_p$, $\gamma_s$ and $\gamma$ are given by
\begin{eqnarray}
\gamma_p &=& \frac{1}{2} \ln\left( \frac{\delta_p}{2-\delta_p}\right) \nonumber \\
\gamma_s &=& \frac{1}{2} \ln\left( \frac{\delta_s}{2-\delta_s}\right)\nonumber \\
\gamma &=& \ln\left[ \left(\delta_p(2-\delta_p)\right)^{n_p/2}\left(\delta_s(2-\delta_s)\right)^{n_l/2} \right]\;.\nonumber
\end{eqnarray}
Notice that the transfer matrix is a product of $B_p(z_S)$ and $A_s(z_T^*)$ but the Hamiltonian is a linear combination of operators $B_p^0$ and $A_s^0$.
Since $0 \leq \delta_{s,p} \leq 1$, it is not difficult to see that $~\gamma_{s,p} \leq 0$ and $~ \gamma \leq 0$. 
This Hamiltonian is the one we got from a deformed $3-$manifold invariant, its partition function ($Z=e^{-\Delta t ~H}$) is not topological for all values of the parameters $\gamma_p$ and $\gamma_s$, but in the limit $\gamma_{s,p} \rightarrow -\infty$ ($\delta_{s,p}\rightarrow 0^{+}$) we get the Hamiltonian of Kitaev's model and its partition function becomes topological. Such a partition function is the one which comes from the non-deformed Kuperberg invariant. The constant $\gamma$ is a finite constant, for $0 \leq \delta_{s,p} \leq 1$, and it represents just a shift in the energy levels. From now on we will ignore this term in the Hamiltonian.

In the following sections we are going to present a ground state in terms of diagrams and we are also going to study the degeneracy of the ground state.

\section{Ground States}
\label{sec-groundstate}
~

The Hamiltonian defined in (\ref{hamiltoniankitaev}) is made of a sum of commuting operators, so it is enough to look at the spectrum of each operator separately. 

The operators $A_s^0$ and $B_p^0$ are projectors, which implies  they have the same spectrum which is $spec\left(A_s^0\right)=spec\left(B_p^0\right)=\{0,~n \}$\footnote{Where $n=\textrm{dim}(\mathcal{A})$.}. Consider $\{\vert \Psi_i \rangle \}$ a basis of eigenvectors of $A_s^0$ and $B_p^0$, it is not difficult to see that the Hamiltonian $H(\gamma_p,\gamma_s)$ has the lowest eigenvalue when 
\begin{eqnarray}
B_p^0\vert \Psi_i \rangle &=& n~\vert \Psi_i \rangle \;\;\;\;\;\;\;\; \hbox{for all plaquettes}\label{eigenstatesB} \\
A_s^0 \vert \Psi_i \rangle &=& n~\vert \Psi_i \rangle \;\;\;\;\;\;\;\; \hbox{for all vertexes} \label{eigenstatesA}
\end{eqnarray}
therefore the lowest eigenvalue of $H(\gamma_p,\gamma_s)$ is $\left(-n_p\gamma_p - n_s\gamma_s\right)$. Thus the eigenstates of $H(\gamma_p,\gamma_s)$ are the states $\vert \Psi \rangle$ such that the equations (\ref{eigenstatesB}) and (\ref{eigenstatesA}) holds or equivalently if the equation below holds, 
\begin{equation}
H(\gamma_p,\gamma_s)\vert \Psi \rangle = \left(-n_p\gamma_p - n_s\gamma_s\right)\vert \Psi \rangle. \label{groundstateH}
\end{equation}

The statement in equation (\ref{groundstateH}) gives us the energy eigenvalue of the ground state but it says nothing about the structure of the ground state. In this section we are going to construct two tensor networks that are exact ground states of (\ref{hamiltoniankitaev}). 

Before we proceed, we need to establish the behaviour of plaquette and vertex operators under change of curve orientation. One can show  that
$B_p(z)$ and $A_v(z^*)$ are mapped to $B_p(S(z))$ and $A_v(S(z^*))$. A proof for this statement can be found in appendix \ref{ap-orientation}.

The plaquette operator we have previously defined is the one shown in figure \ref{plaquette-orientation-a} for $z_S=\eta$. But it can also be written as the ones shown in figures \ref{plaquette-orientation-b} and \ref{plaquette-orientation-c}. 
This change of orientation will not affect the coloured operators provided that $S(z_S)=z_S$.
\begin{figure}[h!]
\centering
\subfigure[]{
\includegraphics[scale=1]{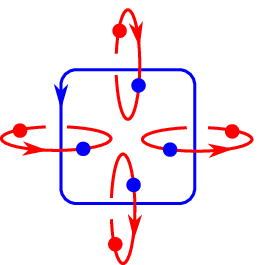} \label{plaquette-orientation-a}
}
\hspace{1.5cm}
\subfigure[]{
\includegraphics[scale=1]{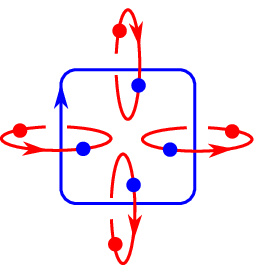} \label{plaquette-orientation-b}
}
\hspace{1.5cm}
\subfigure[]{
\includegraphics[scale=1]{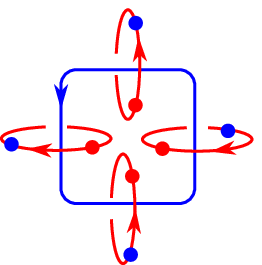} \label{plaquette-orientation-c}
}
\caption{Possibles curves orientation for the plaquette operator.}
\label{plaquette-orientation}
\end{figure}
The same thing happens for the star operator. In figure \ref{star-orientation} all the operators are equal despite the differences in orientation. 
\begin{figure}[h!]
\centering
\subfigure[]{
\includegraphics[scale=1]{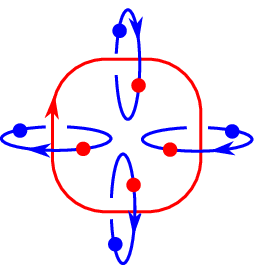} \label{star-orientation-a}
}
\hspace{1.5cm}
\subfigure[]{
\includegraphics[scale=1]{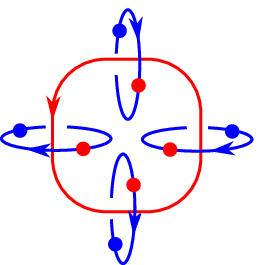} \label{star-orientation-b}
}
\hspace{1.5cm}
\subfigure[]{
\includegraphics[scale=1]{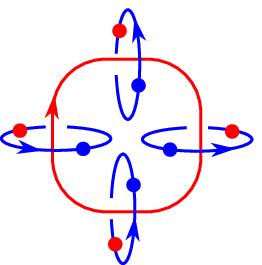} \label{star-orientation-c}
}
\caption{Possibles curves orientation for the star operator.}
\label{star-orientation}
\end{figure}

\subsection{Constructing $|\psi_p\rangle$}
~

The way we will construct a ground state is by taking the product of all plaquette operators (the one drawn in figure \ref{plaquette-orientation-c}), and then feeding the input arrows with the co-integral. In other words the ground state will be of the form
$$\vert \Psi_p \rangle = \prod_p B_p^0\left(\vert\lambda \rangle\otimes \vert\lambda \rangle\otimes\cdots \otimes \vert\lambda \rangle \right)\;,$$
of course it is a eigenstate for all the plaquette operators, since an operator $B_{p^\prime}$ commutes with all $B_p$'s, which means that
\begin{equation}
B_{p^\prime}\prod_p B_p^0 = n \prod_p B_p^0 \;\; \Rightarrow \;\; B_{p^\prime} \vert \Psi_p \rangle = n \vert \Psi_p \rangle\;.\label{B-eigenstate}
\end{equation}
It remains to show that it is also an eigenstate of the star operator. 

In the following we are going to shown a diagrammatic way to represent such a state. First we draw a diagram for the product of all plaquette operators, as shown in figure \ref{buildingpsip}, then we feed all the red dots in figure \ref{buildingpsip} with the co-integral (which is the trace in the regular representation) and after that we get the state drawn in figure \ref{buildingpsipII}.
\begin{figure}[h!]
	\begin{center}
		\includegraphics[scale=1]{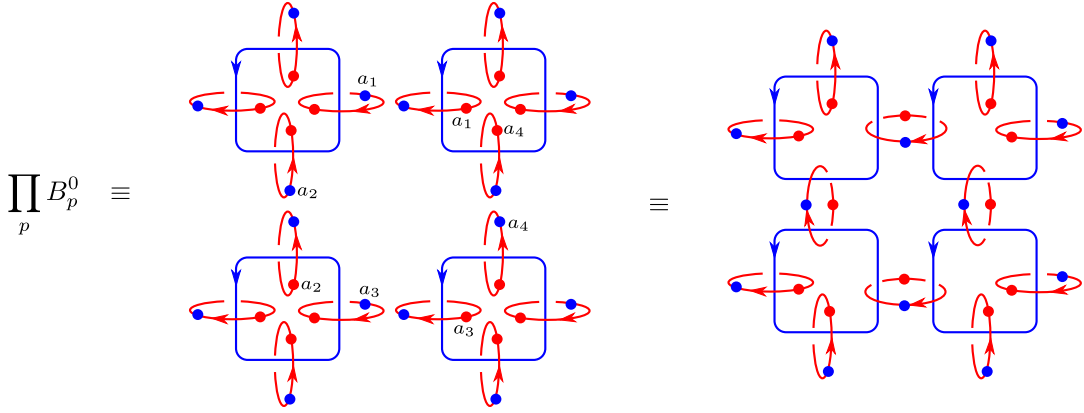}
	\caption{Product of all the plaquette operators.}	
	\label{buildingpsip}
	\end{center}
\end{figure}
\begin{figure}[h!]
	\begin{center}
		\includegraphics[scale=1]{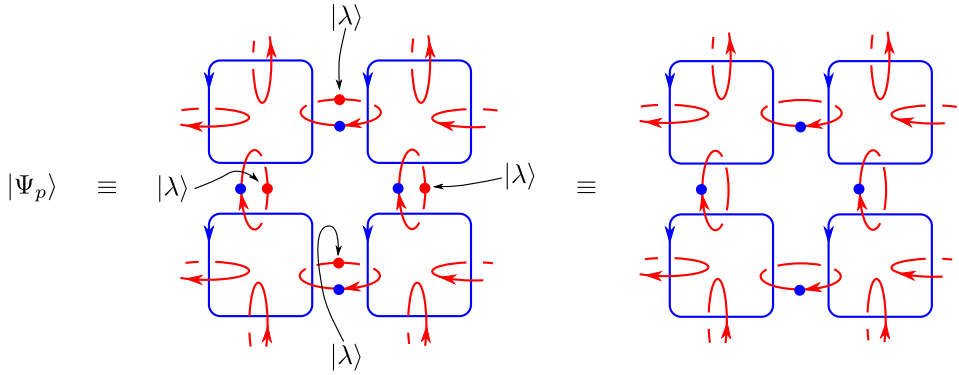}
	\caption{The ground state built from the product of plaquette operators.}	
	\label{buildingpsipII}
	\end{center}
\end{figure}
Note that the state drawn in figure \ref{buildingpsipII} has one blue curve for each plaquette and one red curve for each link, as illustrated in figure \ref{states-a}, where we removed the orientation of the curves for simplicity.
\begin{figure}[h!]
	\begin{center}
		\includegraphics[scale=1]{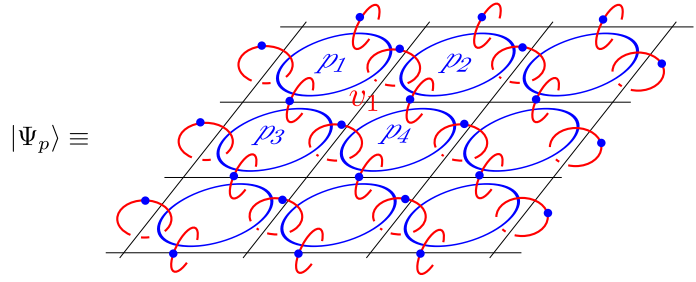}
	\caption{The ground state build from the product of plaquette operators, it has one blue curve for each plaquette and one red curve for each link.}	
	\label{states-a}
	\end{center}
\end{figure}

We are going to use the curves diagram to show that
the state drawn in figure \ref{states-a} is a ground state for the Hamiltonian (\ref{hamiltoniankitaev}). Let us start by deriving equation (\ref{B-eigenstate}) one more. The result of applying a plaquette operator on such a state is shown on figure \ref{plaquettegroundstate-a}.
\begin{figure}[h!]
	\begin{center}
		\includegraphics[scale=1]{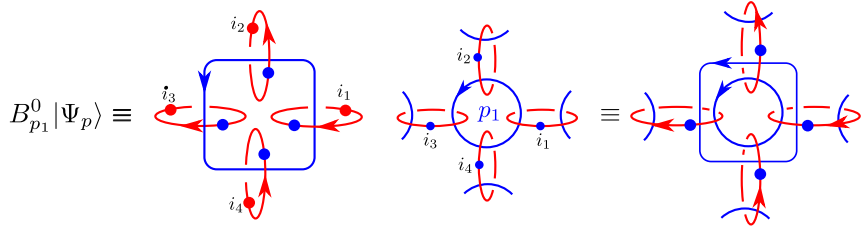}
	\caption{Applying the plaquette operator on a plaquette $p_1$.}	
	\label{plaquettegroundstate-a}
	\end{center}
\end{figure}
We end up in the diagram in figure \ref{plaquettegroundstate-b}
after removing the outside blue curve using the sliding and the two-point moves described in appendix \ref{ap-moves}.  
\begin{figure}[h!]
	\begin{center}
		\includegraphics[scale=1]{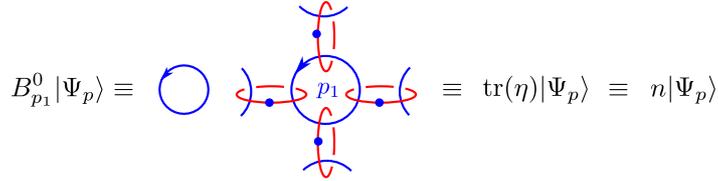}
	\caption{Here we can see that $\vert \Psi_p \rangle$ is an eigenstate of the plaquette operator with eigenvalue $n$.}	
	\label{plaquettegroundstate-b}
	\end{center}
\end{figure}
Exactly the same computation holds for all the plaquettes.

We still have to prove that $\vert \Psi_p \rangle$ is also an eigenstate of the star operator. For that let us apply the star operator on the vertex $v_1$, as described by figure \ref{stargroundstate-a}. In this case we will have to use a slightly different procedure to remove the red curve in the center of the diagram. Instead of one slide move in the plaquette operator case, we will have to use four slide moves, one over each red curve associated to the vertex $v_1$, see figure \ref{stargroundstate-b}. The same procedure holds for all vertex operators. Thus these computations show us that the state $\vert \Psi_p \rangle$ (figure \ref{states-a}) is in fact one ground state of the Hamiltonian (\ref{hamiltoniankitaev}).
\begin{figure}[h!]
	\begin{center}
		\includegraphics[scale=1]{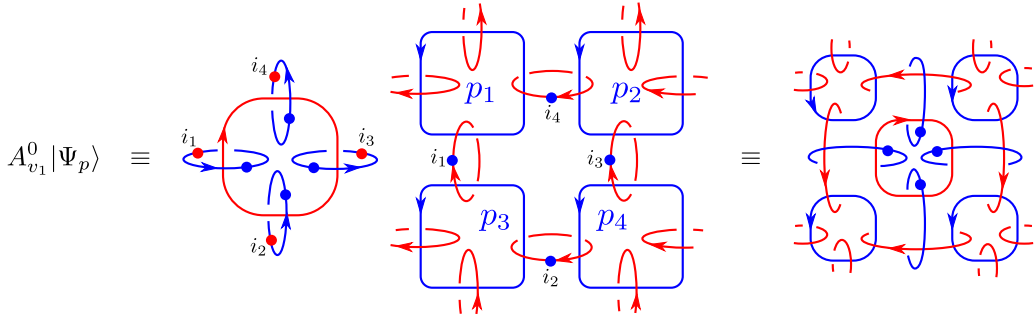}
	\caption{The star operator acting on a vertex $v_1$ of the lattice.}	
	\label{stargroundstate-a}
	\end{center}
\end{figure}
\begin{figure}[h!]
	\begin{center}
		\includegraphics[scale=1]{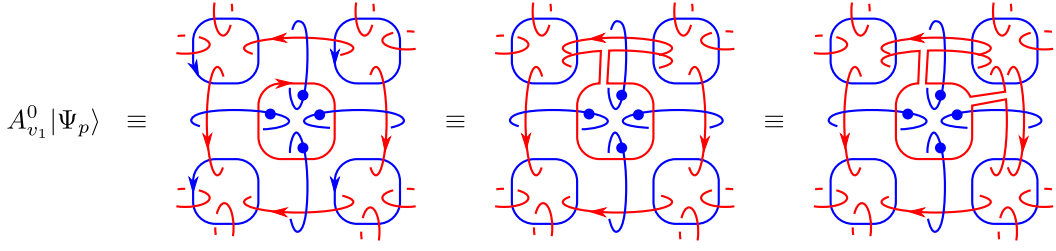}
	\caption{The star operator acting on a vertex $v_1$ of the lattice, after three slides of red curves. In this figure we are not drawing the blue dots in the blue curves just for simplicity.}	
	\label{stargroundstate-b}
	\end{center}
\end{figure}
\begin{figure}[h!]
	\begin{center}
		\includegraphics[scale=1]{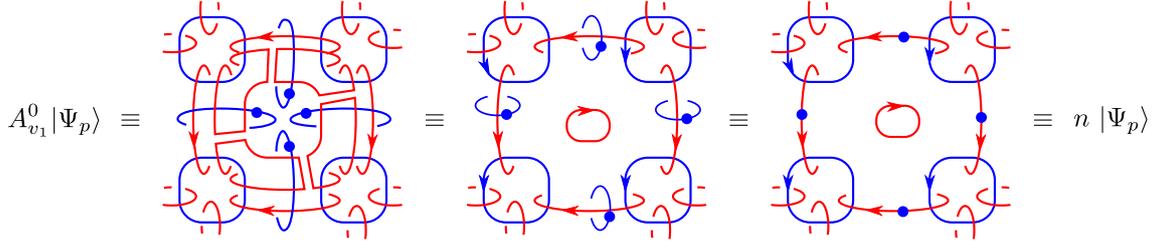}
	\caption{After four slides of red curves we can use two-point moves to get rid of some crossings.}	
	\label{stargroundstate-c}
	\end{center}
\end{figure}

Figure \ref{states-a} is a graphical representation of a tensor network. The corresponding Kuperberg diagram is given in figure \ref{desenho3}. 

\begin{figure}[h!]
	\begin{center}
		\includegraphics[scale=1]{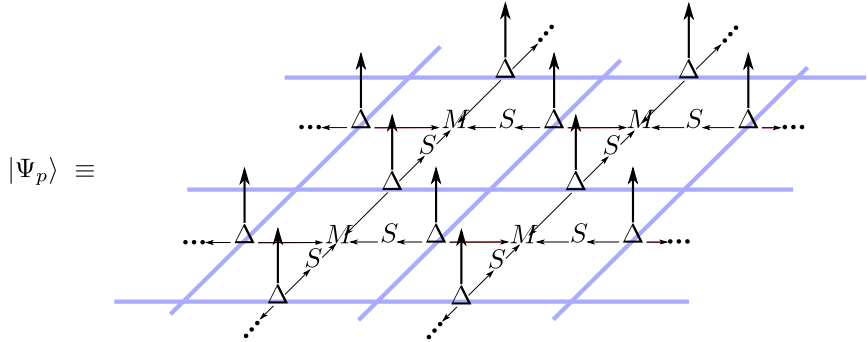}
	\caption{State $\vert\Psi_p\rangle$ written as a tensor network in the Kuperberg notation.} 
	\label{desenho3}
	\end{center}
\end{figure}

\subsection{Constructing $|\psi_s\rangle$}
~

The ground state we have just discussed was made of product of all plaquette operators with all its input arrows fed with co-integrals. Using the same line of thought one can write a ground state made with a product of vertex operators fed with some element of the algebra. In fact, there is another ground state $\vert \Psi_s \rangle$ we can get by feeding the product of vertex operators with a tensorial product of the unit of the algebra (such a state was also written down in \cite{AG2}). That is
$$\vert \Psi_s \rangle = \prod_s A_s^0 \left(\vert \eta \rangle \otimes\vert \eta \rangle \otimes \cdots \otimes \vert \eta \rangle \right)\;. $$
This new state is for sure an eigenstate of the vertex operator, due to the fact that any vertex operator $A_{s^\prime}$ commutes with all the others and the fact that $A_{s^\prime}$ is a projector
\begin{equation}
A_{s^\prime}\prod_s A_s^0= n ~\prod_s A_s^0 \;\; \Rightarrow \;\; A_{s^\prime}\vert \Psi_s\rangle=n\vert \Psi_s\rangle \;. \label{A-eigensta}
\end{equation}
The product of all the star operators, in terms of diagrams, is the one shown in figure \ref{buildingpsis}. In order to get the state we just attach one unit element in each free red dot, as shown in figure \ref{buildingpsisII}.
\begin{figure}[h!]
	\begin{center}
		\includegraphics[scale=1]{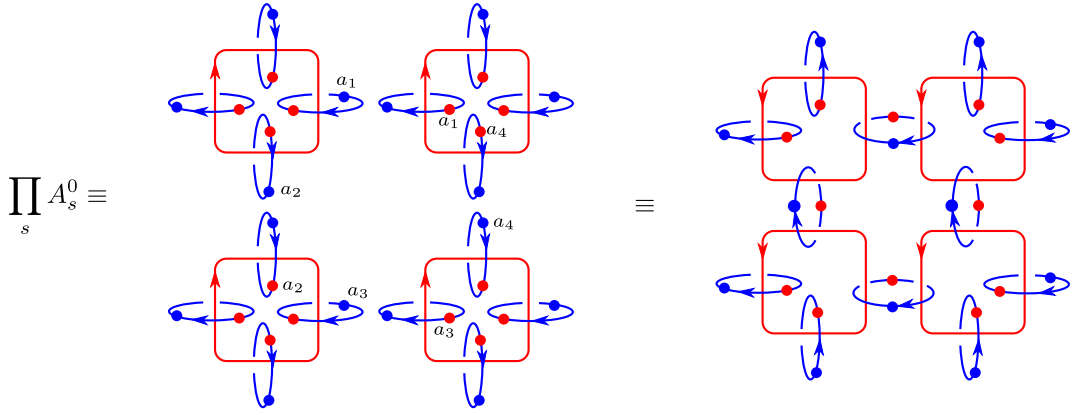}
	\caption{Product of all the vertex operators.}	
	\label{buildingpsis}
	\end{center}
\end{figure}
\begin{figure}[h!]
	\begin{center}
		\includegraphics[scale=1]{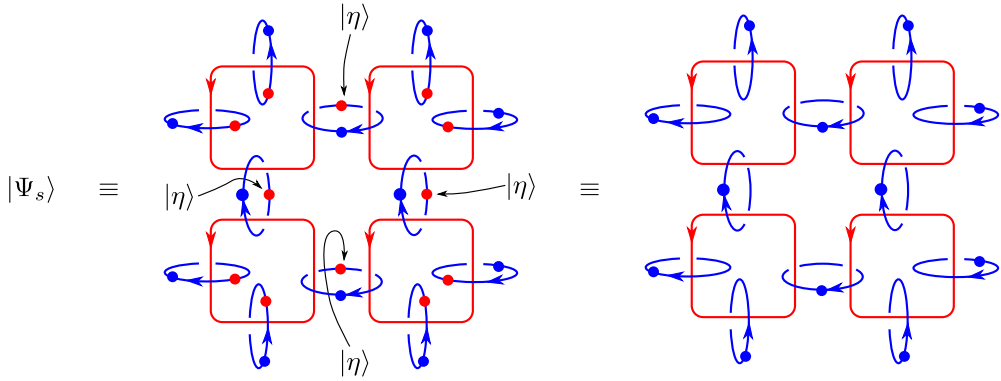}
	\caption{The ground state built from the product of vertex operators.}	
	\label{buildingpsisII}
	\end{center}
\end{figure}

A graphical proof that $\vert \Psi_s \rangle$ is a eigenstate with eigenvalue equal to $n$ for both the plaquette and vertex operators is completely analogous to the one presented for $\vert \Psi_p \rangle$. All steps are the same if we exchange blue and red curves.

As for the ground state $\vert \Psi_p \rangle$, the ground state $\vert \Psi_s \rangle$ has a geometrical meaning. It has one red curve for each link of the lattice and one blue curve for each link, as illustrated in figure \ref{statesII-a}. In the figure we have not drawn the orientations for simplicity. The corresponding tensor network is given in figure \ref{desenho4}.
\begin{figure}[h!]
	\begin{center}
		\includegraphics[scale=1]{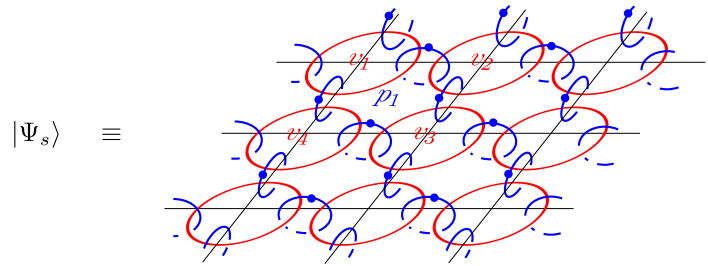}
	\caption{The ground state built from the product of vertex operators. It has one red curve for each vertex and one blue curve for each link.}	
	\label{statesII-a}
	\end{center}
\end{figure}
\begin{figure}[h!]
	\begin{center}
		\includegraphics[scale=1]{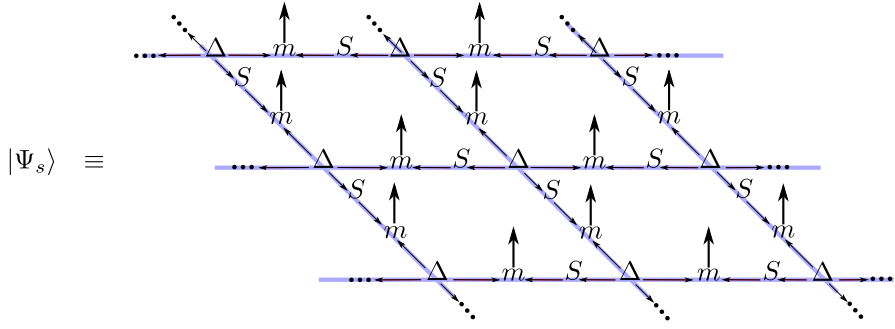}
	\caption{State $\vert\Psi_s\rangle$ written as a tensor network in the Kuperberg notation.} 
	\label{desenho4}
	\end{center}
\end{figure}

\subsection{Comparing $|\Psi_p\rangle$ and $|\Psi_s\rangle$}
~

We can understand the difference between the ground states $|\Psi_p\rangle$ and $|\Psi_s\rangle$ by considering them in the loop representation~\cite{loop}. As $|\Psi_p\rangle$ is obtained by acting with the product of the plaquette operators on the tensor product of co-integrals on all the links, it is precisely the sum of all closed loops. This is because the plaquette operator projects just these states which contribute to the ground state. If we are working on a surface with non-zero genus then this contains the sum of terms belonging to the different topological sectors. Consider for example the toric code $\mathcal{A}=\mathbb{C}(\mathbb{Z}_2)$ on the 2-torus $T^2$.  We have
\begin{equation}
|\Psi_p\rangle = |\chi_1\rangle + |\chi_2\rangle + |\chi_3\rangle + |\chi_4\rangle
\end{equation}
where $|\chi_i\rangle$ are the ground states belonging to the different sectors distinguished by the different non-contractible loops of $T^2$.

On the other hand while obtaining $|\Psi_s\rangle$ we started with a configuration which chose a particular sector, namely $|\chi_1\rangle$, and then acted with the product of gauge transformations, that is the product of the star operators which do not mix the different topological sectors. Thus on $T^2$ we have
\begin{equation}
|\Psi_s\rangle = |\chi_1\rangle.
\end{equation}   

These two states coincide on the 2-sphere, $S^2$ as there are no non-contractible loops on $S^2$.

The loop representation for the ground states is only true for the abelian quantum double models as the excitations in these models have abelian fusion rules. In the case of non-abelian models the fusion rules are more complicated leading to a representation for the ground states not just in terms of closed loops but also other structures depending on the non-abelian group.
\subsection{Ground State Degeneracy}
~

Consider $\mathcal{\beta}=\{\vert \Psi \rangle \}$ a basis of eigenstates of both $A_s^0$ and $B_p^0$. The ground state space $V$ is defined by
$$V = \{\vert \Psi_0 \rangle: B_p^0\vert \Psi_0 \rangle=A_s^0\vert \Psi_0 \rangle=n~\vert \Psi_0 \rangle\;, \; \forall_{s,p}\}~.$$
In this section we are going to compute the ground state degeneracy of (\ref{hamiltoniankitaev}) defined on a surface $\Sigma$. We will show that $\textrm{dim}(V)$ is determined by the 3-manifold invariant defined in \cite{kuperberg} computed on $\Sigma\times S^1$.

Let us denote by $U^0$ the transfer matrtix (\ref{evolutionoperator}) computed at $z_S= \eta, Z^*_T= \epsilon$. Then
$$ U^0 \vert \Psi \rangle = \left\{
\begin{array}{ccc}
n^{n_p+n_l+n_v}~\vert \Psi \rangle & \hbox{ if } & \vert \Psi \rangle \in V \\
0  & \hbox{ if } & \vert \Psi \rangle \notin V
\end{array}
\right.$$
which means that, in the basis $\mathcal{\beta}$, the operator $U^0$ is of the form
$$U^0 \equiv n^{n_p+n_l+n_v}~\left(
\begin{array}{cc}
\mathbf{1}_{k \times k} & \mathbf{0} \\
\mathbf{0} & \mathbf{0}
\end{array} \right)$$
where $\mathbf{0}$ in the null matrix and $k=\textrm{dim}(V)$. If we take the trace of $U^0$ we get
\begin{equation}
{tr}~(U^0)=n^{n_p+n_l+n_v}~\textrm{dim}(V) \;\; \Rightarrow \;\; \textrm{dim}(V)=\frac{\textrm{tr}~(U^0)}{n^{n_p+n_l+n_v}}\;.
\label{degdeg}
\end{equation}
A similar answer for the ground state degeneracy in terms of the trace of a transfer matrix was obtained in \cite{twist}.

But $\textrm{tr}~(U^0)$ is the partition function $Z(\Sigma\times S^1)$ computed for $z=\eta, z^*=\epsilon$ and 
$$
Z\left( \Sigma \times S^1\right)=\mathcal{F}\left( \Sigma \times S^1\right)n^{g_H+N_B+N_R}.
$$ Here $\mathcal{F}\left( \Sigma \times S^1\right)$ is the Kuperberg invariant, $g_H$ is the genus of the Heegaard spliting, $N_B$ and $N_R$ are the numbers of blue and red curves, respectively \cite{kuperberg}. Therefore 
\begin{equation}
\textrm{dim}(V)= \mathcal{F}\left( \Sigma \times S^1\right) \frac{n^{-g_H+N_B+N_R}}{n^{n_p+n_l+n_v}}\;.\label{dimV-a}
\end{equation}
The factor multiplying $\mathcal{F}$ seems to depend on the size of the lattice. However that is not the case. Looking at figure \ref{slice} we can see that $N_B$ and $N_R$ are related to the parameters $n_p$, $n_l$ and $n_v$ of the $2-$manifold $\Sigma$, by
$$ N_B = n_p + n_l \;\;\;\; , \;\;\;\; N_R = n_v+ n_l ~.$$
Besides, the parameters $n_p$, $n_l$ and $n_v$ satisfy the Euler equation
$$n_p+n_v-n_l = \chi ~,$$
where $\chi=2-2g$ is the Euler caracteristic of $\Sigma$. This way we get
$$\textrm{dim}(V) = \mathcal{F}\left( \Sigma \times S^1\right)~n^{n_p+n_v-g_H-\chi}\;.$$
Now the last step is to write down $g_H$ as a function of $n_p$, $n_v$ and $n_l$. Let us do this looking at the dual Heegaard diagram. In figure \ref{genusheegaard-b} we can see the dual Heegaard diagram of $\Sigma\times I$, shown in figure \ref{genusheegaard-a}. Each grey rod highlighted (figure \ref{genusheegaard-b}) represents a spacelike plaquette.
\begin{figure}[h!]
\centering
\subfigure[]{
\includegraphics[scale=1]{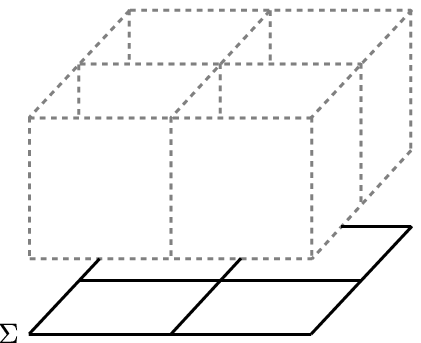} \label{genusheegaard-a}
}
\hspace{.5cm}
\subfigure[]{
\includegraphics[scale=1]{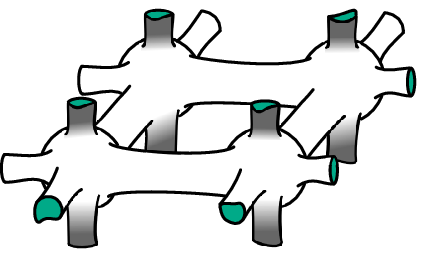} \label{genusheegaard-b}
}
\hspace{.5cm}
\subfigure[]{
\includegraphics[scale=1]{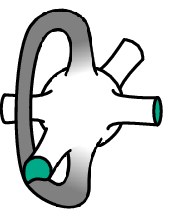} \label{genusheegaard-c}
}
\caption{This figure illustrate how to relate $g_H$ with the parameters of $\Sigma$.}
\label{genusheegaard}
\end{figure}
Note that there is one handle on the Heegaard diagram for each vertex of $\Sigma$, see figure \ref{genusheegaard-a}. Here we are essentially counting the number of non contractible loops and since $\Sigma$ has genus $g$, it has $2g$ more non-contractible loops.  Besides, since we took the trace of the transfer matrix, each plaquette is actually glued to itself, as shown in figure \ref{genusheegaard-c}, and it shows us that each plaquette increases the number of handles by one unit as well. Therefore we can write
$$g_H=n_p+n_v-1+2g$$
where the factor $-1$ is due to the fact we can always stretch one face over around the entire diagram, decreasing the genus by one unit. Therefore the ground state degeneracy becomes
\begin{equation}
\textrm{dim}(V) = \frac{1}{\textrm{dim}(\mathcal{A})}~\mathcal{F}\left( \Sigma \times S^1\right)
\end{equation}
which is clearly a topological invariant.

\subsection{The ground state degeneracy of $\mathbb{C}(\mathbb{Z}_2)$ quantum double models}
Quatum double models with $\mathcal{A}=\mathbb{C}(\mathbb{Z}_2)$ have its ground state degeneracy known~\cite{AK} to be given by
\begin{equation}
\hbox{dim}(V)=2^{2g}\;.
\label{eq-deg}
\end{equation}
Such a result can be obtained from the expression we got for the ground state degeneracy as a special case. In this case the plaquette and vertex operators can be written in terms of the Pauli matrices in the following way
\begin{eqnarray}
B_p & = & \mathbf{1}\otimes\mathbf{1}\otimes\mathbf{1}\otimes\mathbf{1}+ \sigma^z\otimes\sigma^z\otimes\sigma^z\otimes\sigma^z\;, \nonumber \\
A_s & = & \mathbf{1}\otimes\mathbf{1}\otimes\mathbf{1}\otimes\mathbf{1}+ \sigma^x\otimes\sigma^x\otimes\sigma^x\otimes\sigma^x\;.\nonumber
\end{eqnarray}
The degeneracy, according to equation (\ref{degdeg}) is given by
$$\textrm{dim}(V)=\frac{\textrm{tr}(U)}{2^{n_p+n_l+n_v}} \;\;\; \Rightarrow \;\;\; \textrm{dim}(V)=\frac{2^{n_l}~\textrm{tr}\left(\prod_p B_p \prod_s A_s \right)}{2^{n_p+n_l+n_v}}$$
which means we just have to evaluate the trace of the product of all the plaquette and vertex operators. Both plaquette and vertex operators are made of a tensorial product of identities all over the place except in four specific positions which correspond to the link where they are acting, in these specific positions instead of an identity map they have an operator $\sigma^i$ ($i=x$ for the vertex operator and $i=z$ for the plaquette operator) acting on it. It is not difficult to convince ourselves that the product of all plaquette (vertex) operator will have exactly two terms which is a tensorial product of identity maps all over the place. All the remaining terms will have an operator $\sigma^z$ ($\sigma^x$) in at least one position of such a tensorial product. Let us write this statement in the following way
\begin{eqnarray}
\prod_p B_p &=& 2\left(\mathbf{1} \otimes \cdots \otimes \mathbf{1}\right) + f(\cdots \otimes \sigma^z \otimes \cdots)\nonumber \\
\prod_s A_s &=& 2\left(\mathbf{1} \otimes \cdots \otimes \mathbf{1}\right) + f(\cdots \otimes \sigma^x \otimes \cdots)\nonumber \;,
\end{eqnarray}
where $f(\cdots \otimes \sigma^i \otimes \cdots)$ represent the sum of all the remaining terms, which are not a tensorial product of only identity maps.
Using the same argument we can say that the product of $\prod B_p$ with $\prod A_s$ will be of the form
$$\prod_p B_p \prod_s A_s = 4\left(\mathbf{1} \otimes \cdots \otimes \mathbf{1}\right) + f(\cdots \otimes \sigma^i \otimes \cdots)\;.$$
Since $\textrm{tr}(\sigma^i)=0$ and $\textrm{tr}(A\otimes B) = \textrm{tr}(A)\textrm{tr}(B)$ we have 
$$\textrm{tr}\left( \prod_p B_p \prod_s A_s \right) = 4~\left(\textrm{tr}(\mathbf{1})\right)^{n_l}=4~2^{n_l}.$$
Thus the ground state degeneracy is 
$$\hbox{dim}(V)=\frac{4~2^{2 n_l}}{2^{n_p+n_l+n_v}}=4~2^{-(n_p+n_v-n_l)}\;,$$
where $n_p+n_v-n_l$ is the Euler characteristic $\chi=2-2g$, then we can also write the ground state degeneracy as
$$\hbox{dim}(V)=2^{2g}\;,$$
which is the same degeneracy found in \cite{AK}. For abelian groups the same argument holds, with difference that the plaquette and vertex operators will be written in terms of some other operators which play the same role of $\sigma^x$ and $\sigma^z$. For abelian groups the ground state degeneracy will be given by
$$\hbox{dim}(V)=\vert G \vert^{2g}\;.$$

For the non-abelian case it is not possible to derive a general formula as given above. We can however work it out for specific cases. The answer for the case of $\mathbb{C}(S_3)$ is found to be 8 in~\cite{lsnab}. 

\section{Beyond Quantum Double Models}
\label{sec-beyond}
~

The quantum double Hamiltonians came from a partition function $Z\left(\mathcal{A}, z_T, z_S, z_T^*, z_S^*\right)$ for the case when $z_T= \eta$ and $z_S^*=\epsilon$ and $z_T$ and $z_S^*$ are the ones given by equations \ref{zs} and \ref{zt*}.  This leads to the natural question as to what kind of models are obtained for other choices of these parameters. We present a short summary of the results, which are analyzed in detail in an accompanying paper~\cite{p2}.
 
 We divide the summary into two parts. The first of which considers parameter choices which lead to models similar to the quantum double model considered so far in this paper. We explain the manner in which they are similar. The second part deals with parameter choices resulting in perturbations to the quantum double model. These perturbations can be thought of as magnetic fields acting on the links in the $\mathbb{C}(\mathbb{Z}_2)$ QDM, that is the toric code. For the other cases they provide analogous perturbations but with a less clearer meaning. 
 
 Finally we also mention the models obtained when we mix the parameters in the above two cases. 
 
\subsection{Other Quantum Double-like Models}
~
 
To obtain these models we make the following choices for the parameters: we colored timelike link with
 
\begin{equation}
z_T^* = \sum_R\sum_j\sum_{g_j\in C_j}~\chi_R(g_j)~\alpha_j~\psi^{g_j}
\end{equation}
 where $j$ and $R$ run over the conjugacy classes and the irreducible representations of the group $G$, respectively, and $C_j$ denotes the $j$th conjugacy class. For spacelike plaquette coloring we choose
 
 \begin{equation} 
 z_S = \sum_i~\beta_i~\sum_{g\in C_i}~\phi_g
  \end{equation}
 where $i$ runs over the different conjugacy classes.  For timelike plaquettes and spacelike links we choose the trivial coloringm given by
 
 \begin{equation} 
 z_T = \eta~\textrm{and}~z_S^* = \epsilon.
 \end{equation}
The Hamiltonians we obtain are given by \footnote{sl=spacelike link, tl=timelike link, sp=spacelike plaquette and tp=timelike plaquette.}
 \begin{equation} 
 H^{\textrm{sp},~\textrm{tl}} = \sum_p\sum_i~\textrm{ln}(\beta_i)~B_p^{C_i} + \sum_s\sum_R~A_s^R(z_T^*) 
 \end{equation}
where $B_p^{C_i}$ are the projectors to the conjugacy classes (flux sectors) $C_i$ and $A_s^R(z_T^*)$ are projectors to the irreducible representation (charge sector) $R$. In this picture the QDM discussed earlier in this paper is made up of projectors to the vacuum flux and charge sectors of the chosen algebra. It should be noted that the ground states of these models can be thought of as excited states for the ``vacuum sector'' QDM. A similar consideration was also made recently by~\cite{lsnab}. The analysis of this model including the tensor network representation of the ground states is discussed in detail in~\cite{p2}.
 
 \subsection{Perturbed QDMs}
 ~

 These models are obtained with the following parameter choices
 
 \begin{equation} 
 z_T = \sum_i~\alpha_i~\sum_{g\in C_i}~\phi_g
  \end{equation}
 
 \begin{equation} 
 z_S^* = \sum_i~\beta_i~\sum_{g\in C_i}~\psi^{g} 
  \end{equation}
  where $i$ runs over the different conjugacy classes. 
  
 \begin{equation}
 z_S = \eta~\textrm{and}~z_T^* = \epsilon.
 \end{equation}
 
 The Hamiltonian is given by 
 
\begin{equation}
H^{\textrm{sl},~\textrm{tp}} = \sum_p~\left(B_p-\mathbf{1}\right) + \sum_s~\left(A_s - \mathbf{1}\right) + \sum_s~\textrm{ln}(L_s) + \sum_p~\textrm{ln}(T_p) + \sum_{s',p'}~\textrm{ln}(L_s'T_p').
\end{equation}
 
 The details of the operators $L_s$ and $T_p$ are explained in~\cite{p2}. To give a picture of what these models look like we consider the familiar $\mathbb{C}(\mathbb{Z}_2)$ case.
 
 \begin{equation}
  H^{\textrm{sl},~\textrm{tp}} = \sum_p~\left(B_p-\mathbf{1}\right) + \sum_s~\left(A_s - \mathbf{1}\right) + \alpha_x\sum_l~\sigma^x + \alpha_y\sum_l~\sigma^y + \alpha_z\sum_l~\sigma^z.
  \end{equation}
 
 The detailed computations of the ground states of these models are shown in~\cite{p2}.
 
 The most general model is obtained by choosing a non-trivial central element for all the four parts of the 3D lattice.

\subsection{Other Phases}
~

The parameter space explored so far was restricted to the region which produced the toric code/quantum double phase. The ground states of the Hamiltonians describing this phase were seen to be topological. However the transfer matrices built using Kuperberg invariants give us the freedom to study other regions of the parameter space which belong to other phases which are not topological in the sense of the quantum double models. This freedom is due to the fact that the method used so far is a general way of studying a lattice gauge theory based on discrete groups. Thus it allows us to probe the entire dynamics of these theories through the choice of parameters in the Hamiltonians. 

In the following we will elaborate these statements by showing other distinct phases including the well known short-ranged ferromagnetic phase, a non-topological phase and other topological phases between this non-topological phase and the toric code phase. These phases are exhibited for the case when the algebra is $\mathbb{C}(\mathbb{Z}_2)$. Generalization to other algebras is straightforward and so we do not show them here.

\subsubsection{The Ferromagnetic Phase}

This phase can be obtained from the transfer matrix in the following way. The fully dynamical transfer matrix is given by 
\begin{equation}
 U = \prod_v\left(\alpha~A_v + \alpha^{\perp}~A_v^{\perp}\right)~\prod_p\left(\beta~B_p + \beta^{\perp}~B_p^{\perp}\right),
 \end{equation}
where $\alpha$, $\alpha^{\perp}$, $\beta$ and $\beta^{\perp}$ are the parameters used to color the timelike links and the spacelike plaquettes. For example in the case of the involutory Hopf algebra being the group algebra of $\mathbb{Z}_2$ we have the spacelike plaquette and timelike link colors given by
\begin{equation} 
z_S = \beta~\phi_0 + \beta^{\perp}~\phi_1
\end{equation}
and
\begin{equation}
z^*_T = \frac{(\alpha+\alpha^{\perp})}{2}~\psi^0 + \frac{(\alpha-\alpha^{\perp})}{2}~\psi^1.
\end{equation}

By choosing the color of the spacelike plaquettes to be $z_S = \phi_0 + \phi_1$ we remove the effect of all plaquette operators as they become the identity operator with this choice. We then choose $z_T^* = \frac{e^k+\frac{1}{e^k}}{2}~\psi^0 + \frac{e^k-\frac{1}{e^k}}{2}~\psi^1$ as the color for the timelike links to obtain the Hamiltonian as a sum of the vertex operators given by 
\begin{equation}
A_v - A_v^{\perp} = k~ \left[\sigma^x\otimes \sigma^x\otimes \sigma^x\otimes \sigma^x\right]
 \end{equation}
which makes the Hamiltonian 
\begin{equation} 
H = k\sum_v ~ \sigma^x\otimes \sigma^x\otimes \sigma^x\otimes \sigma^x
 \end{equation}
where the Pauli matrices act on the links surrounding a vertex $v$. This Hamiltonian has a product state as its ground state and thus describes a short-ranged entangled phase. 

It is easy to see that by making the choices $z_S = e^k~\phi_0 -\frac{1}{e^k}~\phi_1$ and $z_T^* =  \psi^0$ we obtain the Hamiltonian
\begin{equation} 
H = k~\sum_p ~ \sigma^z\otimes \sigma^z\otimes \sigma^z\otimes \sigma^z
 \end{equation} 
where the Pauli matrices now act along the edges of a plaquette $p$. This Hamiltonian again as a product state as its ground state and is short-ranged entangled.

\subsubsection{The Non-Topological Phase}

This phase is obtained by choosing $z_S = \phi_0 - \phi_1$ and $z_T^* = \psi^1$. This makes the transfer matrix
\begin{equation}
 U = \prod_v~\left(A_v + i^2~A_v^{\perp}\right)\prod_p~\left(B_p + i^2~B_p^{\perp}\right).
 \end{equation}

The Hamiltonian is made up of operators which act on six edges of a lattice. The action of this operator is shown in the figure \ref{squarelattice}.
\begin{figure}[h!]
	\begin{center}
		\includegraphics[scale=1]{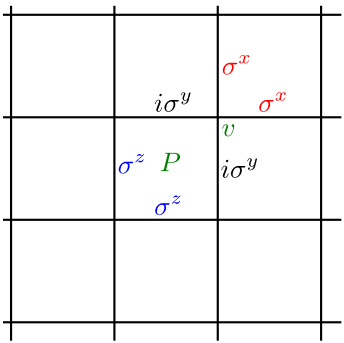}
		\caption{The figure shows the action of the operator $O_p$.}
	\label{squarelattice}
	\end{center}
\end{figure}
\begin{equation}
 H = \sum_p~O_p 
 \end{equation}
where
\begin{equation}
 O_p = \textrm{ln}\left[\left(A_v + i^2~A_v^{\perp}\right)\left(B_p + i^2~B_p^{\perp}\right)\right]
  \end{equation}
This can be simplified to 
\begin{equation}
 O_p = i\pi\left[A_vB_p^{\perp} + A_v^{\perp}B_p\right]
  \end{equation}
which in terms of Pauli matrices is given by
\begin{equation}
\label{AB}
 O_p = i\pi\left[ \mathbf{1}\otimes\mathbf{1}\otimes\mathbf{1}\otimes\mathbf{1}\otimes\mathbf{1}\otimes\mathbf{1} - \sigma^z\otimes\sigma^z\otimes\sigma^y\otimes\sigma^y\otimes\sigma^x\otimes\sigma^x\right]. 
 \end{equation}


The ground states of this model can be obtained by noting that $O_p$ can be thought of as the product of the vertex and plaquette operators in the toric code.
 \begin{equation}
 \label{AB1}
  O_p = \prod_{i\in\partial p}\otimes_i\sigma^z_i~\prod_{j\in v} \otimes_j\sigma^x_j\equiv \prod_p\tilde{B_p}\prod_v\tilde{A_v}.
  \end{equation}
 Note that we have omitted the identity term from Eq.(\ref{AB}) as this just adds a constant term to the Hamiltonian and thus constitutes a constant energy shift. 
 
 Given this form of the operator $O_p$ we can now write down the ground state conditions which include the ones occurring for the toric code model as well.
 The ground state condition is given by
 \begin{equation} 
 O_p|G\rangle  = |G\rangle.
 \end{equation}
 This can be satisfied in two ways for the operator $O_p$ given by Eq.(\ref{AB1}) which are
 \begin{eqnarray} 
 \tilde{B_p}|G\rangle & = & \tilde{A_v}|G\rangle = |G\rangle \\
          \tilde{B_p}|G\rangle & = & \tilde{A_v}|G\rangle = -|G\rangle.
           \end{eqnarray}
The first condition implies that the ground states of this system contains the ground states of the toric code. 
The second condition implies that some of the dyonic excitations in the toric code lie in the ground state of this phase and thus are no longer excitations for this system. This increases the degeneracy of the ground state showing us that this phase is indeed different from that of the toric code. We shall henceforth call this phase the $AB$ phase. In particular we know that the ground states of these models cannot be transformed to a product state using local unitary transformations as the ground states include the ground states of the toric code which are not product states. This implies we are in a quantum phase different from the toric code. However the stability of these phases to perturbations require systematic study similar to ones done for the toric code~\cite{cc}.  

Note that dyonic excitations in the toric code have an orientation. That is for a given plaquette on the square lattice there can be four different dyons depending on which of the four adjacent vertex operators is excited along with this plaquette operator. In the way we have defined the $AB$ phase model the dyons that condense into the ground state is clear from figure \ref{squarelattice}. 

Since the number of dyonic excitations depend on the number of plaquettes the degeneracy picks up a factor depending on this number. More precisely we can write the degeneracy as $2^{2g}f(N_p)$ where $g$ is the genus of the surface. 
It is easy to compute $f(N_p)$. Based on a counting argument we find this as $f(N_p)= 2^{N_p}$. 



As the system is sensitive to the details of the two dimensional graph it is defined on, the $AB$ phase is not a topological one like the toric code. However the degeneracy factors into a product of the degeneracy of the toric code, which is a topological invariant, and a factor depending on the size of the lattice. Such theories are called quasi-topological theories.

We can define models which belong to phases different from the $AB$ phase but are still non-topological in the same way the $AB$ phase is non-topological. An example of this model is got by combining a plaquette operator with the four vertex operators surrounding this plaquette. This model is got from coloring the timelike links with $z_T^* = \psi^0$ and the spacelike plaquettes with $z_S = \phi_0 - \phi_1$. This results in the transfer matrix given by
$$ U = \prod_v~\left(A_v + A_v^{\perp}\right)\prod_p~\left(B_p - B_p^{\perp}\right). $$

Each vertex operator is shared by four plaquette operators. This results in the following Hamiltonian
$$ H = \sum_p~Q_p $$
where 
$$ Q_p = \textrm{ln}\left[\left(A_{v_1} + i^2~A_{v_1}^{\perp}\right)\left(A_{v_2} + i^2~A_{v_2}^{\perp}\right)\left(A_{v_3} + i^2~A_{v_3}^{\perp}\right)\left(A_{v_4} + i^2~A_{v_4}^{\perp}\right)\left(B_p + i^2~B_p^{\perp}\right)\right] $$
which gives 
$$ Q_p = i\frac{\pi}{2}\bigotimes_{\partial_p}\sigma^y\bigotimes_{\textrm{legs}\in p}\sigma^x.$$
We have omitted the constant identity term. 

The ground states of this Hamiltonian can be computed as in the previous case. It comprises of the toric code ground states along with all the dyons in the toric code which now condense to the ground state sector in this model. This clearly increases the degeneracy of the ground states.  This phase is again non topological as the degeneracy depends on the number of plaquettes. It is however quasi-topological as the degeneracy is still a product of the degeneracy of the toric code times a factor depending on just the size of the lattice that is the number of plaquettes. It is given by $2^{2g}\times 2^{4 N_p}$. The factor 4 implies that there are no dyonic excitations in this model. 

Note that we have defined this model on the square lattice. If we try writing this model on other lattices, say the hexagonal lattice, we will end up with unpaired vertex operators. This leads to additional terms in the Hamiltonian. However this does not change the $AB$ phase. On a triangular lattice we find unpaired plaquette operators. 

\subsubsection{Other Topological Phases}

We can mix the $AB$ phase with the toric code phase to find phases with degeneracies which do not depend on the lattice size. Let us consider an example of this case arising by just pairing a single plaquette with a vertex operator. This can be got from the following transfer matrix
$$ T = \prod_{p, p\neq p_0}\tilde{B}_p\prod_{v, v\neq v_0} \tilde{A}_v\times \left(B_{p_0}+i^2B^{\perp}_{p_0}\right)\left(A_{v_0}+i^2A^{\perp}_{v_0}\right) $$
where $p_0$ and $v_0$ are the plaquette and the vertex whose operators get paired. The colors for obtaining these operators are the ones which are used to obtain the $AB$ phase discussed earlier. The colors are for the spacelike plaquette $p_0$ and the timelike link $v_0$. $\tilde{B}_p = \beta B_p + \beta^{\perp} B_p^{\perp}$ and $\tilde{A}_v = \alpha A_v + \alpha^{\perp} A_v^{\perp}$ are obtained by coloring the spacelike plaquettes with $z_S = \beta\phi_0 + \beta^{\perp}\phi_1$ and the timelike links with $z_T^* = \left(\frac{\alpha + \alpha^{\perp}}{2}\right)\psi^0 + \left(\frac{\alpha - \alpha^{\perp}}{2}\right)\psi^1$. The Hamiltonian obtained from this transfer matrix is given by 
$$ H = \textrm{ln}(\alpha)\sum_{v, v\neq v_0}A_v + \textrm{ln}(\alpha^{\perp})\sum_{v, v\neq v_0}A_v^{\perp} + \textrm{ln}(\beta)\sum_{p, p\neq p_0}B_p + \textrm{ln}(\beta^{\perp})\sum_{p, p\neq p_0}B_p^{\perp} + i\pi\left[A_{v_0}B_{p_0}^{\perp} + A_{v_0}^{\perp}B_{p_0}\right] .$$

The ground states of this model includes the toric code ground states. It also includes the state with a single dyonic excitation in the pair $(p_0, v_0$. Thus the degeneracy of ground state sector is now $2^{2g}\times 2= 2^{2g+1}$. Thus the degeneracy is independent of the lattice size unlike the $AB$ phase, hence it is topological. It is clear that if we pair more plaquettes and vertices we keep increasing the ground state degeneracy, tending towards the non-topological phase of the $AB$ model.

\section{Discussion}
\label{sec-discussion}
~

Toric codes and their generalizations, quantum double Hamiltonians, have been presented using a spacetime picture.  In order to to achieve our goal, we have used the same data and notation used in the definition of the three dimensional Kuperberg invariant. Though we used a square lattice throughout this paper, the construction can be carried out for a three dimensional lattice which is a product of a two dimensional lattice, $\Sigma$ and the unit interval for any $\Sigma$ with an arbitrary degree for the vertex. The transfer matrices built on these manifolds were split into a product of plaquette and vertex operators which were precisely those given by the plaquette, $B_p$ and star, $A_s$ operators of the undeformed quantum double Hamiltonian~\cite{Aguado}. The corresponding Hamiltonians of these models were obtained from the transfer matrices by taking the logarithm.

The formalism builds a transfer matrix which includes a large part of the parameter space of a lattice gauge theory based on involutory Hopf algebras. Thus we can expect to find phases other than the ones respresented by the quantum double Hamiltonians. This has been illustrated by deriving short-ranged enltangled ferromagnetic phases and the quasi-topological $AB$ phase. This phase also presents a way of condensing excitations of the quantum double Hamiltonian \cite{Bombin}. We plan to explore these phases in future papers.
 
 This formalism can be thought of as an extension of earlier works~\cite{kuperberg} to the quantum case. We believe these methods can be exploited to obtain a variety of one dimensional Hamiltonian models as well. They could also be used to construct edge modes in these models which have been of interest~\cite{KitMaj} in the context of experimentally observing Majorana modes~\cite{Majorana}. This requires introduction of matter fields living in the vertices in the form of modules. This will significantly change the Kuperberg invariant. We have done this extension for several dimensions and these works are under preparation. 
 
 Our formalism provides a convenient way to make computations in Hopf algebras using the slides of red and blue curves as explained in the appendix. It should be noted that this is in contrast with computations using the Sweedler notation~\cite{Sweedler} which can turn out to be quite cumbersome in these situations. A distinct advantage of this method was seen in the computation of the ground states of the unperturbed quantum double models. We wrote down two ground states for these models. One of these, $|\psi_p\rangle$, was obtained in~\cite{Aguado} and the other one, $|\psi_s\rangle$, was obtained in~\cite{AG2}. Here we further explained that state by breaking it down in the language of spins living on the edges of the two dimensional lattice where our systems live. The other ground state was shown to be different from the first one on a surface with non-trivial topology. The explanation for these in terms of the loop representations for the ground states was given for the Abelian case.
 
 The local language we used in this paper is natural for interpreting the states obtained as tensor network representations, more precisely these are the projected entangled pair states(PEPS)~\cite{PEPS1, PEPS2}. The local tensors can be immediately identified in the way we have presented the ground states. These states have bond dimension equal to the dimension of the group algebra considered. We would like to emphasize that these states are exact and involves no approximation in finding the local tensors as done in the traditional methods~\cite{PEPS1, PEPS2}. We however still need to use renormalization group methods to work with these states for computing mean values of physical observables and thus phase diagrams~\cite{fvRG, Cody, LevinRG}. We present such an analysis diagrammatically in a future work.
 
 This setup is convenient to study perturbations to these models in a systematic way. This has been completely analyzed in an accompanying paper~\cite{p2}, where we obtain the deformations to generalized quantum double models. The exact tensor tensor network representations for some of these ground states have also been shown in~\cite{p2}. In~\cite{p2} we essentially study a bigger part of the parameter space. This is done by also ''non-trivially" coloring timelike plaquettes and spacelike links other than coloring just the timelike links and the spacelike plaquettes.  
 
 We have not provided a complete description of the excited states in these models. The excitations for the Abelian toric codes are known to be got using string operators and for non-Abelian case this is done using ribbon operators~\cite{Bombin}. We will present a full picture of these operators in the formalism introduced here in a forthcoming paper.    

\section{Acknowledgements}
~

The authors would like to thank FAPESP for support of this work. We thank the referees of the Journal of Physics A for their valuable suggestions and constructive criticisms which lead to further improvement of the paper. We also thank our collaborators Juan Pablo Ibieta Jimenez and Hudson Kazuo Teramoto Mendonca for valuable discussions. 

\appendix
\section{Hopf Algebras}
\label{ap-hopfalgebras}
~

In this section we will briefly define and discuss about some identities on Hopf algebras, which are used in this paper. This section is not self contained and most of the proofs can be found in \cite{hopfbook}.  

\subsection{Associative Algebras with Unit}
~

Consider $\mathcal{A}$ a vector space over the field $\mathbb{C}$ whose basis is $\{\varphi_i\}_{i=1}^{n}$ and consider also  $\mathcal{A}^{*}$ its dual vector space, whose basis is $\{ \varphi^i\}_{i=1}^{n}$ such that $\varphi^i \left( \varphi_j\right) = \delta_j^i$. 

We call the vector space $\mathcal{A}$ an algebra if there are two linear maps
$${\it m}:\mathcal{A}\otimes \mathcal{A} \rightarrow \mathcal{A} \;\;\;\; \hbox{and} \;\;\;\; \eta:\mathbb{C}\rightarrow \mathcal{A}$$
such that ${\it m}$ is an associative map and $\eta(1)$ is the unit map of the algebra. The maps ${\it m}$ and $\eta$ are called the multiplication map and the unit map, respectively, and they are defined by the structure constants $m_{ij}^k$ and $\eta^l$ in the following way
\begin{eqnarray}
m &:& \varphi_i \otimes \varphi_j \mapsto \varphi_i \varphi_j=m_{ij}^k \varphi_k ,\nonumber \\
\eta &:& 1 \mapsto \eta(1)= \eta^l \varphi_l. \nonumber
\end{eqnarray}

The structure constants completely define the multiplication map and the unit map, and they are such that the multiplication map is an associative map and $\eta(1)$ is the unit of the algebra.
$$\left(\varphi_a \varphi_b \right)\varphi_c = \varphi_a \left(\varphi_b \varphi_c \right) ,$$
in terms of structure constants it means that
\begin{equation}
m_{ab}^i m_{ic}^l = m_{bc}^i m_{ai}^l .
\label{condicaotensormultiplicacao}
\end{equation}

The fact that $\eta(1)$ is an unit map means that
$$ \eta(1)\varphi_i = \varphi_i \eta(1) = \varphi_i ,\;\; \hspace{10px} \forall i,$$
which give us the following

\begin{equation}
\label{condicaotensorunidade}
\eta^l m_{li}^k = \delta_i^k = \eta^l m_{il}^k.
\end{equation}

All these relations can be presented by Kuperberg diagrams, as defined before. In terms of these diagrams the condition (\ref{condicaotensormultiplicacao}) is the one shown in figure \ref{condicaotensormultiplicacao} and the condition (\ref{condicaotensorunidade}) the one shown in figure \ref{condicaotensorunidade}. 

\begin{figure}[h!]
	\begin{center}
		\includegraphics[scale=1]{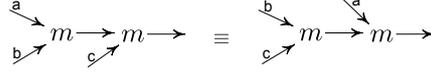}
		\caption{Associativity condition for the multiplication map.}
		\label{condicaotensormultiplicacao}
	\end{center}
\end{figure}

\begin{figure}[h!]
	\begin{center}
		\includegraphics[scale=1]{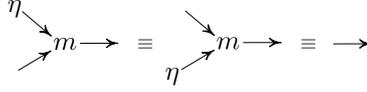}
		\caption{Existence condition of the unit $\eta(1)$ of the algebra.}
	\label{condicaotensorunidade}
	\end{center}
\end{figure}

If the diagrams shown in figures \ref{condicaotensormultiplicacao} and \ref{condicaotensorunidade} are satisfied we say that the trio $\langle \mathcal{A}, {\it m}, \eta \rangle$ define an associative algebra with unit. 

\subsection{Co-associative Co-algebras with Co-unit}
~

We can define in the dual vector space $\mathcal{A}^{*}$ an algebra structure just like we have done in $\mathcal{A}$. So let us define a multiplication map and an unit in the dual space. Alternatively, we can define a co-algebra in $\mathcal{A}$ which has a co-unit. For that consider the trio $\langle \mathcal{A}^{*} , \Delta , \epsilon\rangle$ where $\Delta$ is a linear multiplication map in $\mathcal{A}^{*}$ and $\epsilon$ is its unit, in other words
$$ \Delta: \mathcal{A}^{*} \otimes \mathcal{A}^{*} \rightarrow \mathcal{A}^{*} \;\;\hbox{such that}\;\;\Delta\left( \varphi^i \otimes \varphi^j\right) = \varphi^i \varphi^j =\Delta_k^{ij} \varphi^k, $$
and
$$\epsilon: \mathbb{C} \rightarrow \mathcal{A}^{*} \;\; \hbox{ such that }\;\; \epsilon(1) \;\; \hbox{ is the unit of}\;\;\; \Delta .$$

Since these maps define an algebra, so certainly they satisfy the conditions shown in figure \ref{condicoescoalgebras}. 

\begin{figure}[htb!]
\centering
\subfigure[]{
\includegraphics[scale=1]{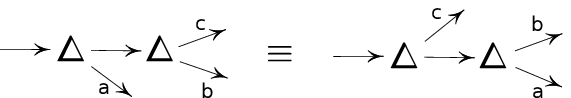} \label{condicoescoalgebras-a}
}
\hspace{2cm}
\subfigure[]{
\includegraphics[scale=1]{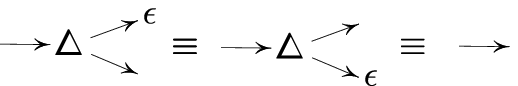} \label{condicoescoalgebras-b}
}
\caption{Condition for the trio $\langle \mathcal{A}^{*},\Delta, \epsilon \rangle$ to be an algebra.}
\label{condicoescoalgebras}
\end{figure}

Once we define an algebra structure in $\mathcal{A}^{*}$ we can think of the maps $\Delta$ and $\epsilon$ as being a co-multiplication and a co-unit in $\mathcal{A}$, respectively. So let the map $\Delta: \mathcal{A} \rightarrow \mathcal{A} \otimes \mathcal{A}$ be defined as
$$\Delta \left( \varphi_i\right) = \Delta_i^{jk} \left( \varphi_j \otimes \varphi_k \right),$$
and let the map $\epsilon: \mathcal{A} \rightarrow \mathbb{C}$ be defined as
$$ \epsilon\left(\varphi_i \right)=\epsilon^i $$
such that the trio $\langle \mathcal{A}^{*}, \Delta, \epsilon \rangle$ define an algebra in $\mathcal{A}^{*}$. We say that the quintet $\langle \mathcal{A}, {\it m},e,\Delta, \epsilon \rangle$  define an algebra and co-algebra in $\mathcal{A}$. As we are going to see later, sometimes there is certain compatibility between the algebra structure and the co-algebra structure, when it happens we say that the quintet defines a bi-algebra.

\subsection{Bi-algebras}
~

As we said before, the quintet $\langle \mathcal{A}, {\it m},\eta,\Delta, \epsilon \rangle$ define an algebra and co-algebra in $\mathcal{A}$ only if the conditions in figure \ref{condicaotensormultiplicacao}, \ref{condicaotensorunidade} and \ref{condicoescoalgebras} are satisfied, but when there is a special kind of compatibility between them we call this a bi-algebra. This compatibility happens when the co-multiplication map and the co-unit map are a homomorphism of the algebra, or in other words, when the following conditions hold 
\begin{eqnarray}
\Delta \left(\varphi_a \varphi_b\right) &=& \Delta \left( \varphi_a\right) \Delta \left(\varphi_b \right), \nonumber \\
\epsilon \left (\varphi_a \varphi_b \right) &=& \epsilon \left ( \varphi_a\right)\epsilon \left ( \varphi_b\right) .\nonumber
\end{eqnarray}

The Kuperberg diagrams for the condition above is the one shown in figure \ref{condicaobialgebra3}. 

\begin{figure}[htb!]
\centering
\subfigure[]{
\includegraphics[scale=1]{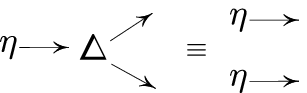} \label{condicaobialgebra3-a}
}
\hspace{1.5cm}
\subfigure[]{
\includegraphics[scale=1]{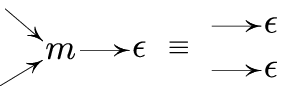} \label{condicaobialgebra3-b}
}
\hspace{1.5cm}
\subfigure[]{
\includegraphics[scale=1]{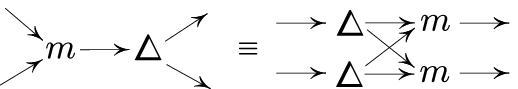} \label{condicaobialgebra3-c}
}
\caption{Bi-algebra condition.}
\label{condicaobialgebra3}
\end{figure}

Thus, the relations shown in figure \ref{condicaobialgebra3} have to be satisfied. Once we define the bi-algebra structure we are ready to define the Hopf algebra structure, but for this purpose  we still need one more ingredient which is an endomorphism $S$.

\subsection{Hopf Algebras}
~

Consider the quintet $\langle \mathcal{A}, {\it m},\eta,\Delta, \epsilon \rangle$ which is a bi-algebra and also consider an endomorphism $S$ of $\mathcal{A}$ which is called the antipode. If $S$ satisfies the antipode axiom, shown in figure\ref{condicaoantipoda}, we say that the sextet $\langle \mathcal{A}, {\it m},\eta,\Delta, \epsilon,S\rangle$ defines a Hopf algebra.
\begin{figure}[h!]
	\begin{center}
		\includegraphics[scale=1]{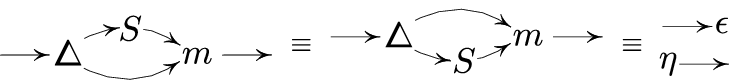}
		\caption{Antipode axiom.}
	\label{condicaoantipoda}
	\end{center}
\end{figure}

Briefly speaking, we can say that the sextet $\langle \mathcal{A}, {\it m},\eta,\Delta, \epsilon,S\rangle$ defines a Hopf algebra if, and only if, all the axioms in figure \ref{todosaxiomas} hold.

\begin{figure}[htb!]
\centering
\subfigure[]{
\includegraphics[scale=1]{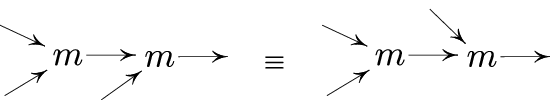} \label{todosaxiomas-a}
}
\hspace{3cm}
\subfigure[]{
\includegraphics[scale=1]{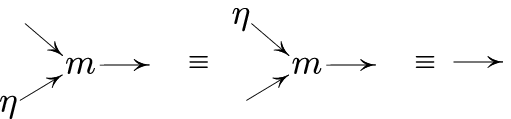} \label{todosaxiomas-b}
}
\\
\vspace{.5cm}
\subfigure[]{
\includegraphics[scale=1]{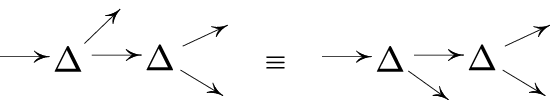} \label{todosaxiomas-c}
}
\hspace{3cm}
\subfigure[]{
\includegraphics[scale=1]{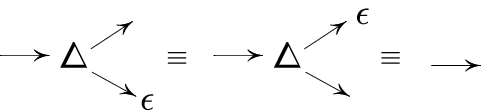} \label{todosaxiomas-d}
}
\\
\vspace{.5cm}
\subfigure[]{
\includegraphics[scale=1]{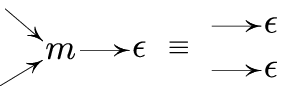} \label{todosaxiomas-e}
}
\hspace{1cm}
\subfigure[]{
\includegraphics[scale=1]{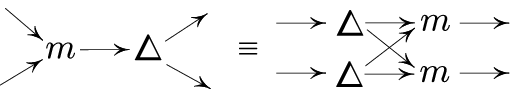} \label{todosaxiomas-f}
}
\hspace{1cm}
\subfigure[]{
\includegraphics[scale=1]{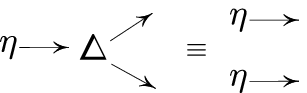} \label{todosaxiomas-g}
}
\\
\vspace{.5cm}
\subfigure[]{
\includegraphics[scale=1]{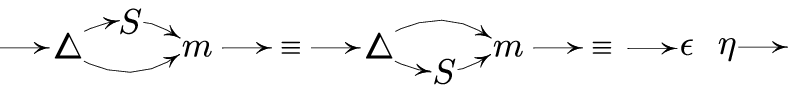} \label{todosaxiomas-h}
}
\caption{Hopf algebra axioms.}
\label{todosaxiomas}
\end{figure}

\subsection{Identities and Properties of Hopf Algebras}
~

Following the same steps shown in \cite{kuperberg} in this section we are going to discuss some of the Hopf algebra identities and properties. In this paper we are considering only involutory Hopf algebras, in other words, Hopf algebras such that $S^2=\mathbb{1}$. All of the results showed below can be found in \cite{hopfbook}.
\begin{prop}\label{prop-trace}
The trace of an element in the regular representation is equal to the tensor $m_{ak}^k$ (figure \ref{proptraco-a}), likewise the co-trace of an element in the dual algebra is equal to the tensor $\Delta_k^{ka}$ (figure \ref{proptraco-b}), in other words
\begin{figure}[htb!]
\centering
\subfigure[]{
\includegraphics[scale=1]{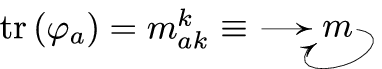} \label{proptraco-a}
}
\hspace{3cm}
\subfigure[]{
\includegraphics[scale=1]{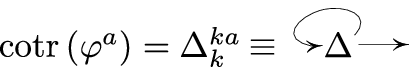} \label{proptraco-b}
}
\caption{Trace and Cotrace in the regular representation.}
\label{proptraco}
\end{figure}
\end{prop}

\begin{proof}
This can be found in \cite{hopfbook}.
\end{proof}

\begin{prop} \label{prop-antihomomorphism}
In a Hopf algebra the antipode is an anti-homomorphism of the (co-)algebra, as shown in figure \ref{identidadeshopfantipoda}.
\begin{figure}[htb!]
\centering
\subfigure[The antipode is an anti-homomorphism of the algebra.]{
\includegraphics[scale=1]{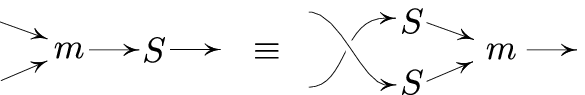} \label{identidadeshopfantipoda-a}
}
\hspace{3cm}
\subfigure[The antipode is an anti-homomorphism of the co-algebra.]{
\includegraphics[scale=1]{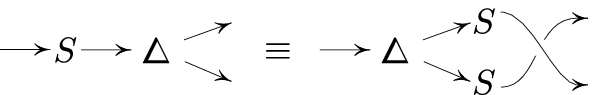} \label{identidadeshopfantipoda-b}
}
\caption{The antipode is an anti-homomorphism of the (co-)algebra.}
\label{identidadeshopfantipoda}
\end{figure}

\end{prop}
\begin{proof}
This can be found in \cite{hopfbook}.
\end{proof}

\begin{prop}\label{prop-unit-antipode}
In a Hopf algebra the unit and the co-unit are invariants under antipode action, as shown in figure \ref{identidadeshopfunidadecounidadeantipoda}.

\begin{figure}[h!]
\centering
\subfigure[]{
\includegraphics[scale=1]{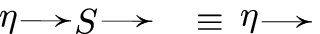} \label{identidadeshopfunidadecounidadeantipoda-a}
}
\hspace{3cm}
\subfigure[]{
\includegraphics[scale=1]{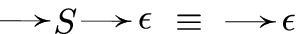} \label{identidadeshopfunidadecounidadeantipoda-b}
}
\caption{The unit and co-unit are invariants under antipode action.}
\label{identidadeshopfunidadecounidadeantipoda}
\end{figure}
\end{prop}

\begin{proof}
The proof for this proposition is quite easy, we just need to use proposition (\ref{prop-antihomomorphism}). In figure \ref{identidadeshopfunidadecounidadeantipodaprova1} we can see how it can be done.

\begin{figure}[h!]
	\begin{center}
		\includegraphics[scale=1]{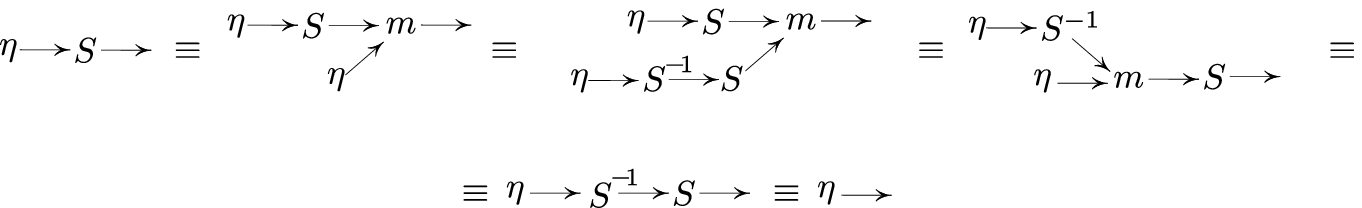}
	\caption{Proof of Propostion 3.}
	\label{identidadeshopfunidadecounidadeantipodaprova1}
	\end{center}
\end{figure}

For the second part of this proof it is enough to just flip all arrows in figure \ref{identidadeshopfunidadecounidadeantipodaprova1} and replace the multiplication tensor by the co-multiplication tensor, and the unit tensor by the co-unit tensor.
\end{proof}

The next results we are going to show are connected either directly or indirectly to the fact that in a Hopf algebra there always exists both an integral and a co-integral (both non-zero).

\begin{definition}\label{defi-integral}
We call a co-integral an element $\lambda \in \mathcal{A}$ of the algebra such that the equality in figure \ref{definicaointegralcointegral-a} holds. Likewise we call an integral an element $\Lambda \in \mathcal{A}^{*}$ in the dual algebra such that the equality in figure \ref{definicaointegralcointegral-b} holds.

\begin{figure}[htb!]
\centering
\subfigure[]{
\includegraphics[scale=1]{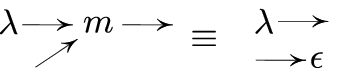} \label{definicaointegralcointegral-a}
}
\hspace{3cm}
\subfigure[]{
\includegraphics[scale=1]{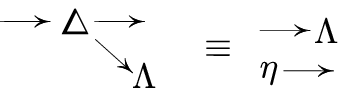} \label{definicaointegralcointegral-b}
}
\caption{In \ref{definicaointegralcointegral-a} the definition of a left co-integral  and in \ref{definicaointegralcointegral-b} the definition of right integral.}	
\label{definicaointegralcointegral}
\end{figure}
\end{definition}

With definition \ref{defi-integral} in mind and the guarantee of existence and uniqueness of (co-)integral in Hopf algebras, provided by \cite{kauffman}, let us consider the following proposition.

\begin{prop} \label{prop-tensor-integral}

In a Hopf algebra the tensor drawn in figure \ref{integraltraco-a} is always an integral of the algebra. But remember we are considering only involutory Hopf algebras, so the figure \ref{integraltraco-a} can be reduced to \ref{integraltraco-b}. Likewise in an involutory Hopf algebra the tensor shown in figure \ref{integraltraco-c} is always a co-integral.

\begin{figure}[h!]
\centering
\subfigure[]{
\includegraphics[scale=1]{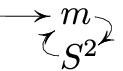} \label{integraltraco-a}
}
\hspace{2cm}
\subfigure[]{
\includegraphics[scale=1]{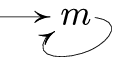} \label{integraltraco-b}
}
\hspace{2cm}
\subfigure[]{
\includegraphics[scale=1]{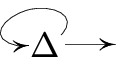} \label{integraltraco-c}
}
\caption{In \ref{integraltraco-a} an integral of the algebra, in \ref{integraltraco-b} the integral of an involutory algebra and in \ref{integraltraco-c} the co-integral of an involutory algebra.}
\label{integraltraco}
\end{figure}

\end{prop}

\begin{proof}
This proof can be found in \cite{hopfbook}.
\end{proof}

\begin{prop}\label{prop-antipode-involutory}

In an involutory Hopf algebra the tensor drawn in figure \ref{valorestabilizacao} is numerically equal to dimension of the algebra.

\begin{figure}[h!]
	\begin{center}
		\includegraphics[scale=1]{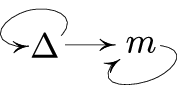}
	\caption{Contraction of an integral and a co-integral.}	
	\label{valorestabilizacao}
	\end{center}
\end{figure}
\end{prop}

\begin{proof}

This proof is trivial and a directly consequence of the fact that the trace is an integral of the algebra, as shown in figure \ref{valorestabilizacaoprova}.

\begin{figure}[h!]
	\begin{center}
		\includegraphics[scale=1]{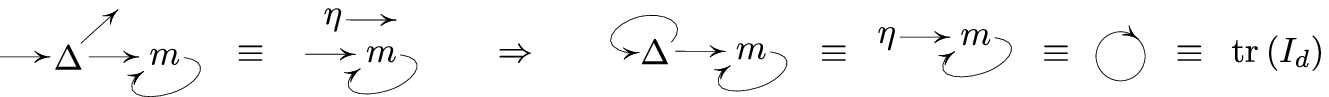}
	\caption{Proof of Proposition (\ref{valorestabilizacao}).}
	\label{valorestabilizacaoprova}
	\end{center}
\end{figure}

So the contraction of the trace with the co-trace is equal to trace of the identity map, therefore it is equal to the dimension of the algebra.

\end{proof}

\begin{prop}
\label{prop-antipodainvoluoria}

If $\mathcal{A}$ is an involutory Hopf algebra the antipode is given by the one shown in figure \ref{antipodainvolutoria}.

\begin{figure}[h!]
	\begin{center}
		\includegraphics[scale=1]{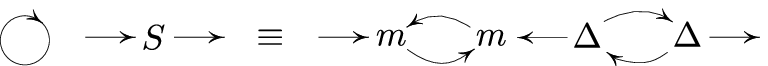}
	\caption{Antipode in an involutory Hopf algebra.}	
	\label{antipodainvolutoria}
	\end{center}
\end{figure}

\end{prop}

\begin{proof}
This proof can be found in \cite{hopfbook}.
\end{proof}

\begin{prop}
\label{prop-orientation}
In a involutory Hopf algebra the tensors in figure \ref{orientationfixed-a} and \ref{orientationfixed-b} are the same of the tensors in figure \ref{orientationfixed-d} and \ref{orientationfixed-c}, respectively, which means they do not depend on the orientation of the curves.
\begin{figure}[h!]
\centering
\subfigure[]{
\includegraphics[scale=1]{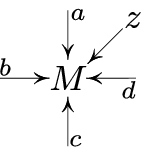} \label{orientationfixed-a}
}
\hspace{2cm}
\subfigure[]{
\includegraphics[scale=1]{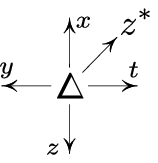} \label{orientationfixed-b}
}
\hspace{2cm}
\subfigure[]{
\includegraphics[scale=1]{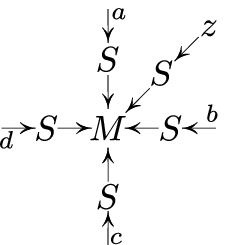} \label{orientationfixed-d}
}
\hspace{2cm}
\subfigure[]{
\includegraphics[scale=1]{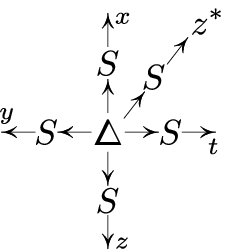} \label{orientationfixed-c}
}
\caption{The tensor $M_{abcd}$ and $\Delta^{xyzt}$ are invariant under multiplication by antipodes.}
\label{orientationfixed}
\end{figure}
\end{prop}

\begin{proof}
This proof is based on the fact that the antipode of any involutory Hopf algebra can be given by the proposition \ref{prop-antipodainvoluoria}. But first note that the tensor in the figure \ref{orientationfixed-d} can be written as shown in figure \ref{orientationfixed-proof-a}, just using the associativity property of the algebra. Then, the next step is to use the fact that the antipode is an anti homeomorphism of the algebra, as also shown in figure \ref{orientationfixed-proof-a}.
\begin{figure}[h!]
\centering
\subfigure[]{
\includegraphics[scale=1]{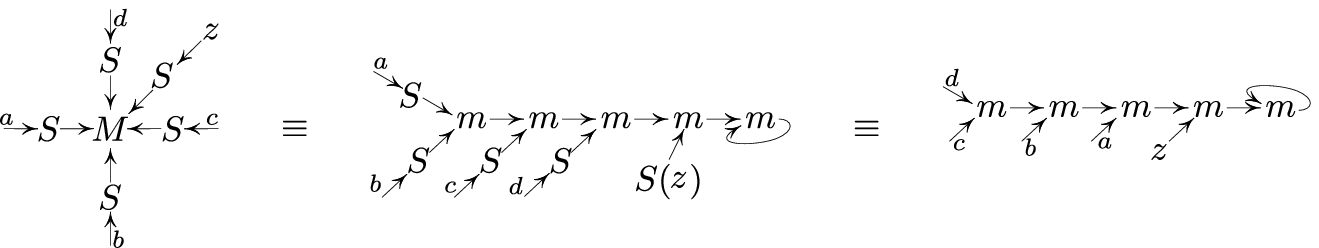} \label{orientationfixed-proof-a}
}
\\
\subfigure[]{
\includegraphics[scale=1]{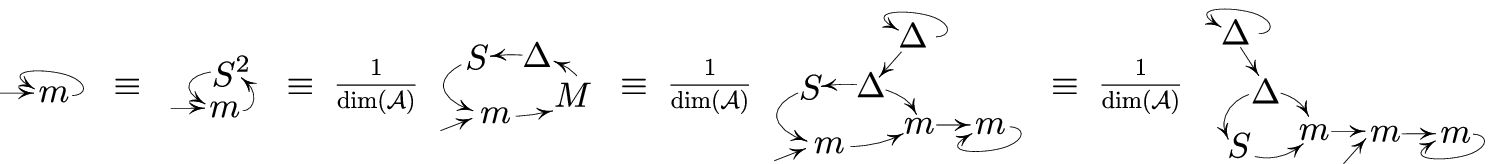} \label{orientationfixed-proof-b}
}
\\
\subfigure[]{
\includegraphics[scale=1]{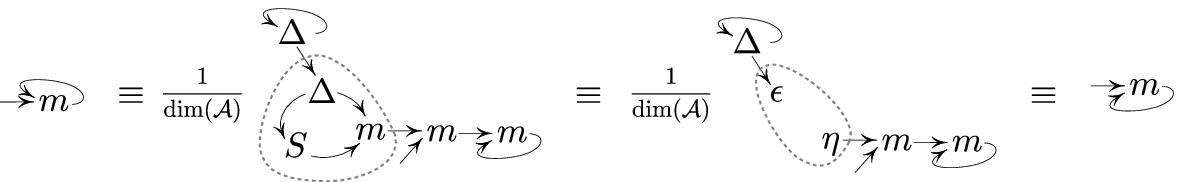} \label{orientationfixed-proof-c}
}
\caption{Proof of the proposition \ref{prop-orientation}.}
\label{orientationfixed2}
\end{figure}
The tensor obtained in figure \ref{orientationfixed-proof-a} is similar to the one shown in figure \ref{orientationfixed-a}, except for two small differences: the trace has been taking on the left, differently of the trace used to define the tensor $M_{abcd}$; the tensor $M_{abcd}$ is now colored with $z$, instead of $S(z)$.  But that is okay, since this two trace are the same for involutory Hopf algebras, as we can see in figures \ref{orientationfixed-proof-b} and \ref{orientationfixed-proof-c}, where we just used the antipode as written in proposition \ref{prop-antipodainvoluoria} and the antipode axiom in order to get the final picture shown in figure \ref{orientationfixed-proof-c}, which proves this proposition. Of course the same procedure can be used to proof that the tensors drawn in figures \ref{orientationfixed-b} and \ref{orientationfixed-c} are the same.
\end{proof}

\section{Some Moves on Diagrams}
\label{ap-moves}
~

In this appendix we are going to define some moves on diagrams, which are the same ones defined in Kuperberg's work \cite{kuperberg}, the validity of such a moves are also shown in \cite{kuperberg}. The first move we will discuss is called the {\it two-point move} and it consists basically of getting rid of crossings between a blue and a red curve which do not contribute to the total weight associated with these curves. The second one is called the {\it slide move} and it consists in sliding a curve over another one with the same color. In the following we will define these two moves. The two-point move holds for any involutory hopf algebra even if the curves are colored by some element of the center of the algebra. On the other hand the slide move does not hold for any element of the center of the algebra.

\subsection{Two-Point Move}
~

The two-point move consists of eliminating two consecutive crossings as shown in figure \ref{twopointmove}. Note that these crossings have different orientations.
\begin{figure}[h!]
	\begin{center}
		\includegraphics[scale=1]{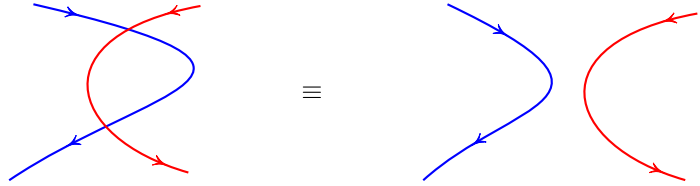}
	\caption{The two-point move.}	
	\label{twopointmove}
	\end{center}
\end{figure}
The weight associated to the diagrams in figure \ref{twopointmove} are the ones shown in figure \ref{twopointmovepesos}
\begin{figure}[h!]
	\begin{center}
		\includegraphics[scale=1]{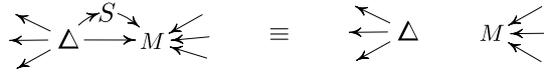}
	\caption{The two-point move weights.}	
	\label{twopointmovepesos}
	\end{center}
\end{figure}
and it is not difficult to see that these two weights are equal, using for that the antipode axiom in figure \ref{todosaxiomas-h}.

\subsection{Slide Move}
~

The second move is the slide move which consists of sliding a curve over another one with the same color. Therefore there are two kinds of slides, the blue ones and the red ones. This move depends on many properties of Hopf algebras. In figure \ref{slidemove} we can see both the red and blues slides.
\begin{figure}[h!]
\centering
\subfigure[Blue Slide.]{
\includegraphics[scale=1]{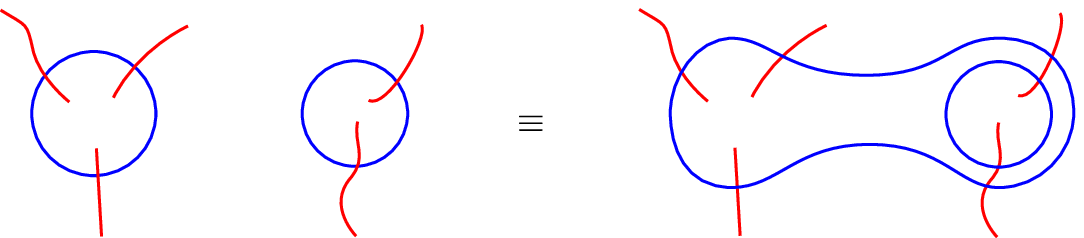} \label{slidemove-a}
}
\\
\subfigure[Red Slide.]{
\includegraphics[scale=1]{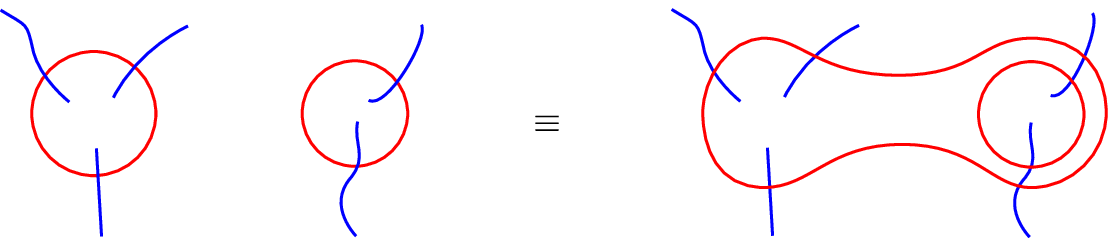} \label{slidemove-b}
}
\caption{The slide move.}
\label{slidemove}
\end{figure}
The Kuperberg diagram associated with the diagram in figure \ref{slidemove-a} is the one shown in figure \ref{slidingprovapesos}, which are actually the same weight \cite{kuperberg}.
\begin{figure}[h!]
	\begin{center}
		\includegraphics[scale=1]{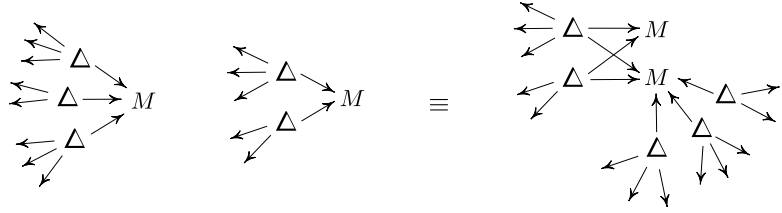}
	\caption{The blue slide move weights.}	
	\label{slidingprovapesos}
	\end{center}
\end{figure}

\section{Powers of the operators $A_s(\epsilon)$ and $B_p(\eta)$}
\label{ap-powers}
~

We can easily compute powers of the star and plaquette operators by doing slides and two-point moves. This procedure simplifies the usual proof using the axioms of Hopf algebra. In the following we will compute these powers only for the plaquette operator. For the star operator it follows exactly in the same way. In terms of diagrams $\left(B_p\right)^2$ is the one shown in figure \ref{powersBp}.
\begin{figure}[h!]
	\begin{center}
		\includegraphics[scale=1]{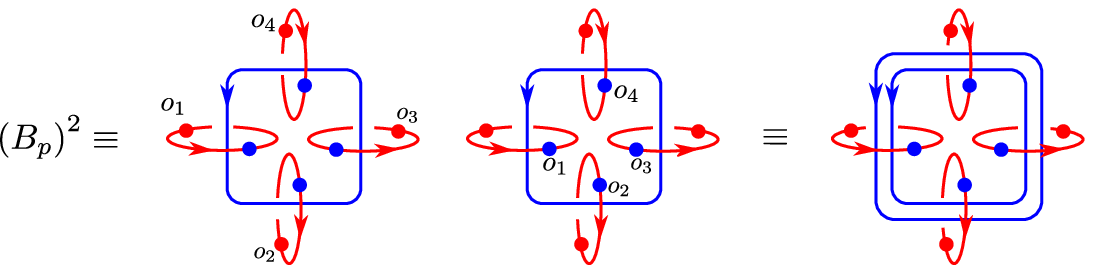}
	\caption{Square of the plaquette operator.}	
	\label{powersBp}
	\end{center}
\end{figure}
In figure \ref{powersBp} we can perform a slide for the external blue blue curve over the internal one and some two-point moves and then we will get the diagram in figure \ref{powersBpslided}.
\begin{figure}[h!]
	\begin{center}
		\includegraphics[scale=1]{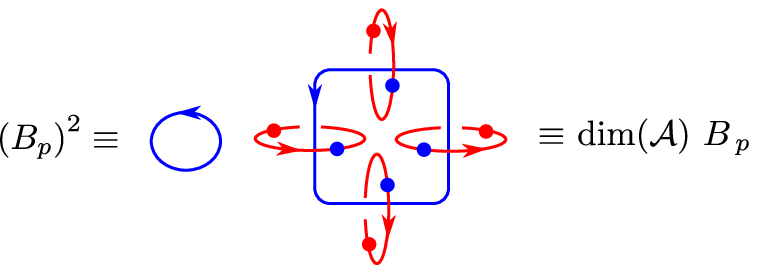}
	\caption{Square of the plaquette operator.}	
	\label{powersBpslided}
	\end{center}
\end{figure}
Repeating the same procedure several times we get the following relations for powers of the plaquette operator (which is the same for the star operator)
\begin{eqnarray}
\left(B_p\right)^k &=& \left(\textrm{dim}(\mathcal{A})\right)^{k-1}~B_p\;,\;\; \hbox{for} \;\; k\neq 0~, \nonumber \\
\left(A_s\right)^k &=& \left(\textrm{dim}(\mathcal{A})\right)^{k-1}~A_s\;,\;\; \hbox{for} \;\; k\neq 0~. \nonumber
\end{eqnarray}

\section{The operators $B_p(\lambda)$ and $A_s(\Lambda)$}
~

The operators $B_p(\lambda)$ and $A_s(\Lambda)$ are proportional to a tensorial product of identity maps. It is a trivial fact as $\lambda$ and $\Lambda$ are the co-integral and the integral of the algebra. To prove this statement we just need to prove one simple property.
\begin{figure}[h!]
	\begin{center}
		\includegraphics[scale=1]{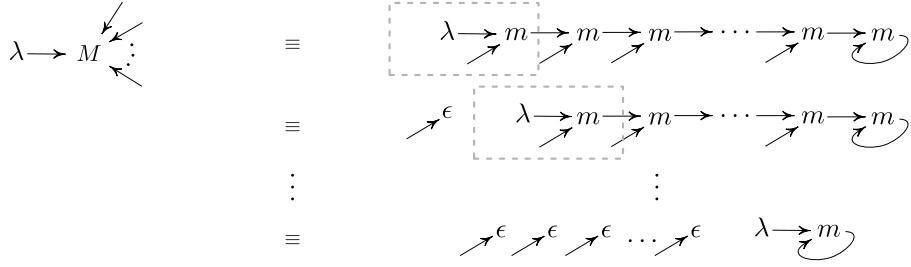}
	\caption{Co-integral property for the multi product.}	
	\label{cointegralgeneralizacao}
	\end{center}
\end{figure}
The plaquette operator $B_p(\lambda)$ becomes the one in figure \ref{plaquetteoperatorintegral}
\begin{figure}[h!]
	\begin{center}
		\includegraphics[scale=1]{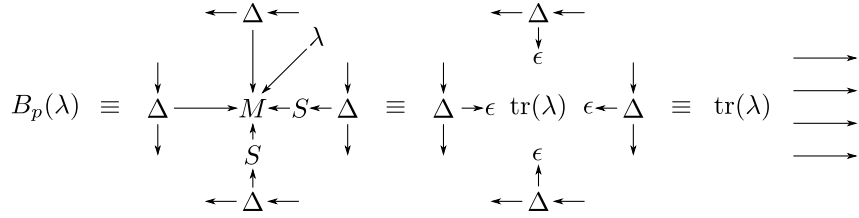}
	\caption{This plaquette operator is proportional to a tensorial product of identities maps.}	
	\label{plaquetteoperatorintegral}
	\end{center}
\end{figure}
or in other words,
\begin{eqnarray}
B_p(\lambda) &=& \textrm{tr}(\lambda) ~\left(\mathbb{1}\otimes\mathbb{1}\otimes\mathbb{1}\otimes\mathbb{1}\right)\; , \nonumber \\
A_s(\Lambda) &=& \textrm{cotr}(\Lambda)~\left( \mathbb{1}\otimes\mathbb{1}\otimes\mathbb{1}\otimes\mathbb{1}\right)\;. \nonumber
\end{eqnarray}

\section{Commutation Relation between the Operators $A_s(z_T^*)$ and $B_p(z_S)$}
\label{ap-commutation}
~

In this section we are going to prove that the star operator commutes with the plaquette operator for any $z_T^*$ and $z_S$ which belongs to concenter and the center of the algebra. We have to analyze two different situations, the first one is the easiest one, where the star and plaquette operator acts on different spaces, as shown in figure \ref{commutator-b}
\begin{figure}[h!]
\centering
\subfigure[]{
\includegraphics[scale=1]{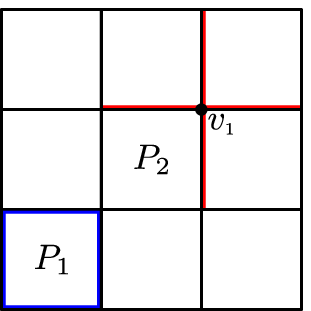} \label{commutator-b}
}
\hspace{4cm}
\subfigure[]{
\includegraphics[scale=1]{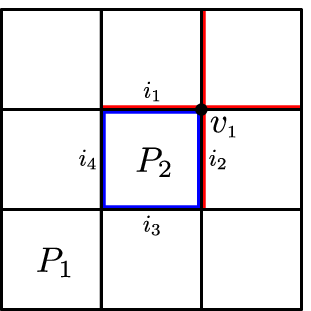} \label{commutator-a}
}
\caption{In (a) the situation where the star and plaquette operator acts on different spaces and in (b) where they do act on the same spaces.}
\label{commutatator-ab}
\end{figure}

Since $A_{v_1}$ and $B_{P_1}$ are acting on different spaces the commutation relation between them is trivial
$$\left[A_{v_1},B_{P_1} \right]=0\; .$$
Now we just have to look at the situation described at figure \ref{commutator-a}. In this case both the star and plaquette operators act on the links $i_1$ and $i_2$. The Kuperberg diagram for these operators are the ones shown in figures \ref{commutator-c} and \ref{commutator-d}.
\begin{figure}[h!]
\centering
\subfigure[]{
\includegraphics[scale=1]{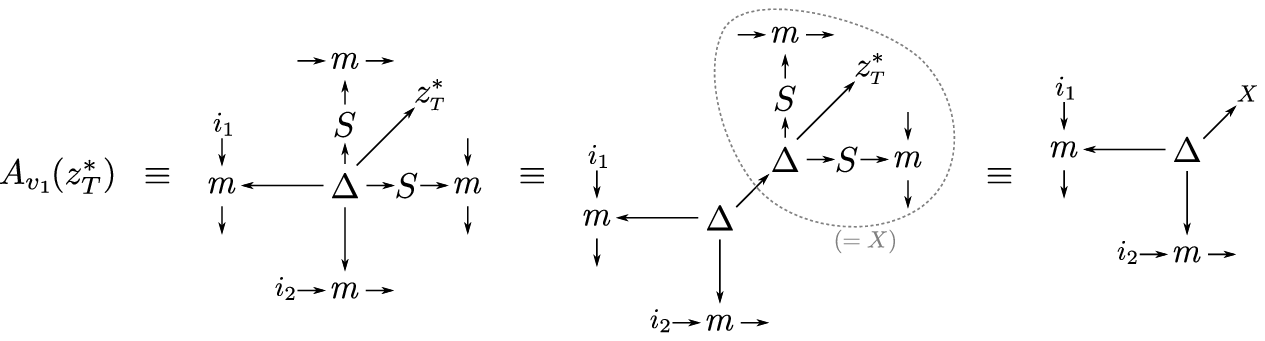} \label{commutator-c}
}
\hspace{4cm}
\subfigure[]{
\includegraphics[scale=1]{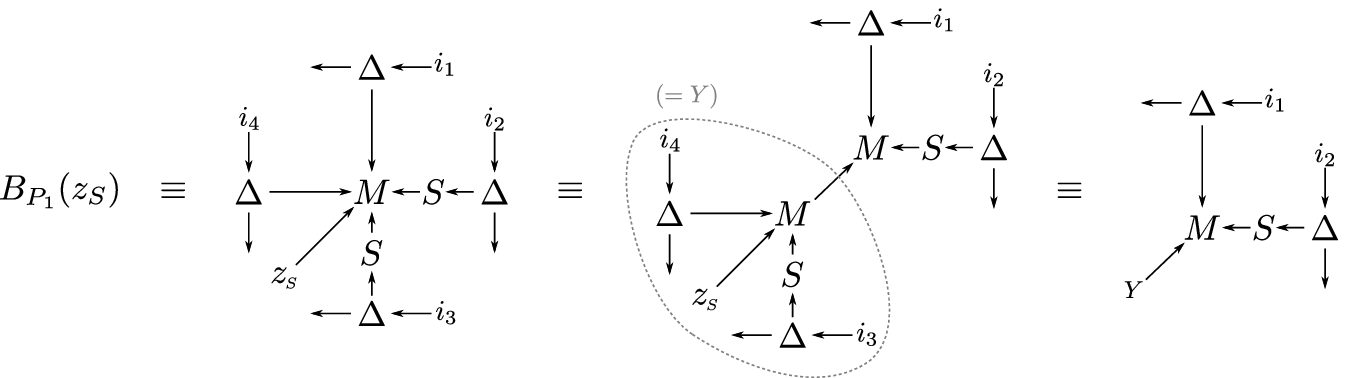} \label{commutator-d}
}
\caption{In the figure $X$ and $Y$ are tensors which represents the whole diagram inside the highlighted areas.}
\label{commutatator-cd}
\end{figure}
As the element $z_S$ ($z_T^*$) belongs to the center (co-center) of the algebra, we can move them to the place which is more convenient, otherwise we would not be able to proceed with this computation, and this operator do not commute in general.

Now we are ready to show that $\left[A_{v_1},B_{P_1}\right]=0$. The way we are going to do this is by showing that $A_{v_1}B_{P_1}=B_{P_1}A_{v_1}$. In figure \ref{commutator-e} we can see the Kuperberg diagram of the product $B_{P_1}A_{v_1}$, and doing some manipulations we end up with the last diagram shown in figure \ref{commutator-e}.
\begin{figure}[h!]
	\begin{center}
		\includegraphics[scale=1]{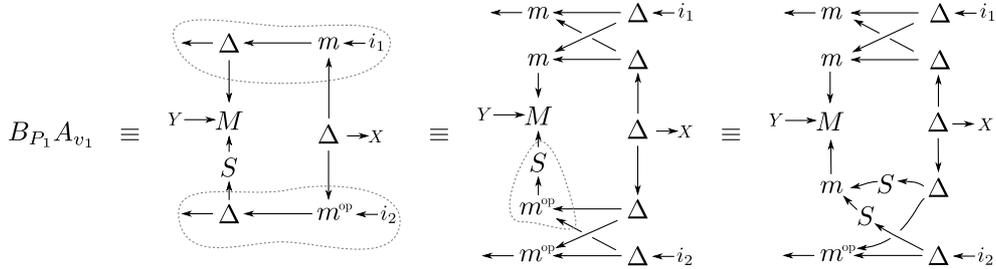}
	\caption{We used the bi-algebra axiom (\ref{todosaxiomas-f}) in the first two highlighted areas and the proposition (\ref{prop-antihomomorphism}) in the second one.}	
	\label{commutator-e}
	\end{center}
\end{figure}

After using the associativity and co-associativity axiom, (\ref{todosaxiomas-a}) and (\ref{todosaxiomas-c}), we get the diagram in figure \ref{commutator-f}.
\begin{figure}[h!]
	\begin{center}
		\includegraphics[scale=1]{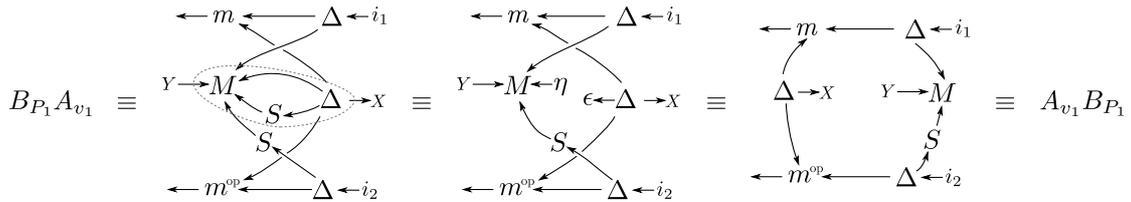}
	\caption{Here we used the antipode axiom (\ref{todosaxiomas-h}) in the highlighted area.}	
	\label{commutator-f}
	\end{center}
\end{figure}

The diagram in figure \ref{commutator-f} is actually the product $A_{v_1}B_{P_1}$. Therefore 
$$\left[A_{v_1},B_{P_1}\right]=0\;.$$

\section{Orientation of the Curves for the Plaquette and the Star Operators}
\label{ap-orientation}
~

The plaquette and star operators are the ones shown in figures \ref{plaquetteoperator-diagram} and \ref{staroperator-diagram}. But these operators can also be written with different choices of curve orientation which are equivalent, in the sense that these different choices will give us the same operators written in a different way. The basic point here is the proposition \ref{prop-orientation}. Using this proposition we can immediately write the two equivalent operators, shown in figure \ref{plaquette-orientation-appendix-a} and \ref{plaquette-orientation-appendix-b}. From figure \ref{plaquette-orientation-appendix-a} to \ref{plaquette-orientation-appendix-b}, we just flipped the orientation of the blue curve (and $z\mapsto S(z)$), which surely do not change the weight associated to such operator, due to proposition  \ref{prop-orientation}. But, equivalently one could have flipped the red curves instead of the blue, it would give us the operator drawn in figure \ref{plaquette-orientation-c}. But, the operator in figure \ref{plaquette-orientation-c} is  equal to the operator drawn in figure \ref{plaquette-orientation-b}, therefore, if $z=S(z)$, all the operators drawn in figure \ref{plaquette-orientation} are equals.

\begin{figure}[h!]
\centering
\subfigure[]{
\includegraphics[scale=1]{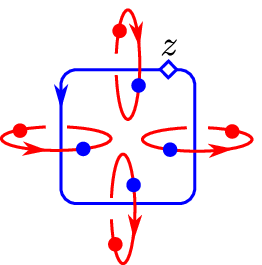} \label{plaquette-orientation-appendix-a}
}
\hspace{1.5cm}
\subfigure[]{
\includegraphics[scale=1]{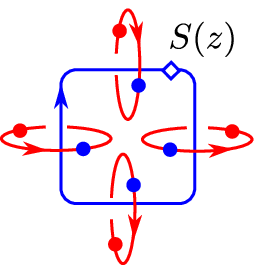} \label{plaquette-orientation-appendix-b}
}
\hspace{1,5cm}
\subfigure[]{
\includegraphics[scale=1]{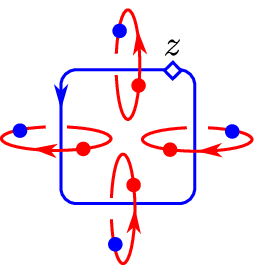} \label{plaquette-orientation-appendix-c}
}
\caption{Different orientations of the blue and red curves for the plaquette operator.}
\label{plaquette-orientation-appendix}
\end{figure}

The same kind of analysis can be done for the vertex operator.

\end{document}